\documentclass{article}
\usepackage{arxiv}

\usepackage[utf8]{inputenc} % allow utf-8 input
\usepackage[T1]{fontenc}    % use 8-bit T1 fonts
\usepackage{hyperref}       % hyperlinks
\usepackage{url}            % simple URL typesetting
\usepackage{booktabs}       % professional-quality tables
\usepackage{amsfonts}       % blackboard math symbols
\usepackage{nicefrac}       % compact symbols for 1/2, etc.
\usepackage{microtype}      % microtypography
\usepackage{lipsum}

\usepackage[small]{caption}
\usepackage{algorithm}
\usepackage{algorithmic}
\usepackage{amssymb}
\usepackage[round]{natbib}
\usepackage{mathtools}
\usepackage{amsmath}
\usepackage{tikz}
\usetikzlibrary{calc}
\usepackage{todonotes}
\usepackage{dsfont}
\usepackage{enumitem}
\usepackage{tikz}
\usepackage{amsthm}
\usepackage{thm-restate}
\usepackage{booktabs}
\usepackage{multirow}
\usepackage{newfloat}
\usepackage{wrapfig}
\usepackage{dsfont}
\usepackage{comment}
\usepackage{multirow}

\newtheorem{example}{Example}
\newtheorem{theorem}{Theorem}
\newtheorem{lemma}{Lemma}

\newtheorem{definition}{Definition}

\newcommand{\avec}{\textbf{a}}

\newcommand{\svec}{\mathbf{s}}
\newcommand{\sset}{\mathcal{S}}
\newcommand{\Expec}{\mathbb{E}}
\newcommand{\defeq}{\vcentcolon=}

\newcommand{\opt}{\textsc{Opt}}
\newcommand{\pr}{\textnormal{Pr}}
\newcommand{\nrec}{\bar n}

\newcommand{\MyParagraph}[1]{\textsc{#1}}

\DeclareMathOperator*{\argmax}{arg\,max}

\newcommand{\br}{\textsf{BR}}

\newcommand{\rec}{\mathcal{R}}
\newcommand{\A}{\mathcal{A}}
\newcommand{\Z}{\mathcal{Z}}

\newcommand{\brset}{\mathcal{M}}
\newcommand{\nAct}{\varrho}

\newcommand{\mUno}{\textsf{Merlin}$_1$}
\newcommand{\mDue}{\textsf{Merlin}$_2$}
\newcommand{\ansUno}{\xi_1}
\newcommand{\ansDue}{\xi_2}
\newcommand{\ansUnoSet}{\Xi_1}
\newcommand{\ansDueSet}{\Xi_2}
\newcommand{\art}{\textsf{Arthur}}
\newcommand{\mStrat}{\eta}
\newcommand{\mStratUno}{\eta_1}
\newcommand{\mStratDue}{\eta_2}
\newcommand{\F}{\mathcal{F}}
\newcommand{\fg}{$\textsc{FreeGame}_\delta$}
\newcommand{\ver}{\mathcal{V}}
\newcommand{\ass}{\zeta}
\newcommand{\twoProvGame}{\mathcal{G}}

\newcommand{\mssi}{$\epsilon$-\textsc{MFS}}
\newcommand{\mssiPrime}{$\epsilon'$-\textsc{MFS}}
\newcommand{\xvec}{\mathbf{x}}

\newcommand{\wvec}{\mathbf{w}}
\newcommand{\vvec}{\mathbf{v}}
\newcommand{\dvec}{\mathbf{d}}
\newcommand{\tcal}{\mathcal{T}}
\newcommand{\ical}{\mathcal{I}}
\newcommand{\pvec}{\mathbf{p}}
\newcommand{\nrow}{n_{\textnormal{row}}}
\newcommand{\ncol}{n_{\textnormal{col}}}
\newcommand{\dist}{\textnormal{dist}}
\newcommand{\type}{\mathtt{t}}
\newcommand{\auxFunc}{\mathfrak{f}}

\newcommand{\K}{\mathcal{K}}
\newcommand{\Pcal}{\mathcal{P}}

\title{Public Bayesian Persuasion: Being Almost Optimal and Almost Persuasive}

\author{
	Matteo~Castiglioni\\
	Politecnico di Milano\\
	\texttt{matteo.castiglioni@polimi.it}
	\And
	Andrea~Celli\\
	Politecnico di Milano\\
	\texttt{andrea.celli@polimi.it}
	\And
	Nicola~Gatti \\
	Politecnico di Milano\\
	\texttt{nicola.gatti@polimi.it} \\
}

\begin{document}

\maketitle

\begin{abstract}
{\em Persuasion} studies how an informed principal may influence the behavior of agents by the strategic provision of payoff-relevant information.
We focus on the fundamental multi-receiver model by~\citet{arieli2019private}, in which there are no inter-agent externalities. 
Unlike prior works on this problem, we study the {\em public} persuasion problem in the general setting with: (i) arbitrary state spaces; (ii) arbitrary action spaces; (iii) arbitrary sender's utility functions.
We fully characterize the computational complexity of computing a bi-criteria approximation of an optimal public signaling scheme.
In particular, we show, in a voting setting of independent interest, that solving this problem requires at least a quasi-polynomial number of steps even in settings with a binary action space, assuming the Exponential Time Hypothesis.
In doing so, we prove that a relaxed version of the \textsc{Maximum Feasible Subsystem of Linear Inequalities} problem requires at least quasi-polynomial time to be solved.
Finally, we close the gap by providing a quasi-polynomial time bi-criteria approximation algorithm for arbitrary public persuasion problems that, in specific settings, yields a QPTAS.
\end{abstract}

\section{Introduction}

{\em Information structure design} studies how to shape agents' beliefs in order to achieve a desired outcome. 
When information is incomplete, the information structure determines ``who knows what'' about the parameters determining payoff functions. 
There has been a recent surge of interest in the study of how an informed principal may influence agents' collective behavior toward a favorable outcome, via the strategic provision of payoff-relevant information.
The prescriptive problems arising in such setting are often termed {\em persuasion} or {\em signaling}.
The study of these problems has been driven by their application in domains such as auctions and online advertisement~\citep{badanidiyuru2018targeting,bro2012send,emek2014signaling}, voting~\citep{alonso2016persuading,cheng2015mixture}, traffic routing~\citep{bhaskar2016hardness,vasserman2015implementing}, recommendation systems~\citep{mansour2016bayesian}, security~\citep{xu2015exploring,xu2016signaling,rabinovich2015information}, and product marketing~\citep{babichenko2017algorithmic,candogan2019persuasion}.

Persuasion is the task faced by an informed principal---the {\em sender}---, trying to influence the behavior of the self-interested agent(s) in the game---the {\em receiver}(s).
Such a sender faces the algorithmic problem of determining the optimal information structure to further her own objectives. 
This is typically modeled through the selection of a {\em signaling scheme}, which maps the sender's parameters observations to distributions over possible signals.
A foundational model describing the persuasion problem is the Bayesian persuasion framework (BP) by~\citet{kamenica2011bayesian}.
Here, there is a sender and a single receiver.
The parameters determining the payoff functions are collectively termed the {\em state of nature}, and model exogenous stochasticity in the environment.
Their prior distribution is known to both the sender and the receiver, but the sender observes the realized state of the environment, originating a fundamental asymmetry in the information available to the two agents.
The prior distribution and the sender's signaling scheme determine the receiver equilibrium behavior.
The model assumes the sender's commitment, which is a natural assumption in many settings~\citep{kamenica2011bayesian,dughmi2017survey}.
One  argument  to  that  effect  is  that  reputation  and credibility may be a key factor for the long-term utility of the sender~\citep{rayo2010optimal}.

In practice, the sender may need to persuade {\em multiple receivers}, revealing information to each of them.
In the multiple-receiver setting, the sender may employ either {\em private} or {\em public} signaling schemes. 
In the former setting, the sender may reveal different information to each receiver through private communication channels.
In the latter, which is the focus of this paper, the sender has to reveal the same information to all receivers.
Public persuasion is well suited for settings where private communication channels are either too costly or impractical (e.g., settings with a large population of receivers, such as voting), and settings where receivers may share private information with each other, which frequently happens in practice.

This paper adopts and generalizes the multi-agent persuasion model by~\citet{arieli2019private}, which rules out the possibility of {\em inter-agent externalities}.
Specifically, each receiver's utility depends only on her own action and the realized state of nature, but not on the actions of other receivers.
This assumption allows one to focus on the key problem of coordinating the receivers' behaviors, without the additional complexity arising from externalities, which have been shown to make the problem largely intractable~\citep{bhaskar2016hardness,rubinstein2015honest}.
Our paper is the first, to the best of our knowledge, focusing on public persuasion with no inter-agent externalities and: (i) an arbitrary space of states of nature; (ii) arbitrary receivers' action spaces; (iii) arbitrary sender's utility function.
Previous works on~\citet{arieli2019private}'s model either address the private persuasion setting~\citep{arieli2019private,babichenko2016computational,dughmi2017algorithmic}, or make some structural assumptions which render them special cases of our model~\citep{xu2019tractability}.

\subsection{Context: Persuasion with Multiple Receivers}

% fare vedere che c'è una serie di lavori che affrontano problemi sempre più generali, e noi li dominiamo strettamente tutti

% dire come i modelli precedenti sono sottocasi dei nostri (vedi dughmi e xu no externalities)

\citet{dughmi2016algorithmic} analyze for the first time Bayesian persuasion from a  computational perspective, focusing on the single receiver case.
\citet{arieli2019private} introduce the model of persuasion with no inter-agent externalities. 
The authors study the setting with binary actions and state spaces, providing a characterization of the optimal signaling scheme in the case of supermodular, anonymous submodular, and supermajority sender's utility functions.
\citet{babichenko2016computational} provide a tight $1-1/e$ approximate signaling scheme for monotone submodular sender's utilities and show that an optimal private scheme for anonymous utility functions can be found efficiently.
\citet{dughmi2017algorithmic} generalize the previous model to the case of many states of nature.

Various works study public persuasion, showing that designing public signaling schemes is usually harder than with private communication channels.
\citet{bhaskar2016hardness} and \citet{rubinstein2015honest} study public signaling problems in which two receivers play a zero-sum game. 
In particular, \citet{bhaskar2016hardness} rule out an additive PTAS assuming the planted-clique hardness. 
Moreover, \citet{rubinstein2015honest} proves that, assuming the Exponential Time Hypothesis (ETH), computing an $\epsilon$-optimal signaling scheme requires at least quasi-polynomial time. 
This result is tight due to the quasi-polynomial approximation scheme by \citet{cheng2015mixture}.

A number of previous works focus on the public signaling problem in the no inter-agent externalities framework by~\citet{arieli2019private}.
\citet{dughmi2017algorithmic} rule out the existence of a PTAS even when receivers have binary action spaces, and objectives are linear, unless $\textsc{P}=\textsc{NP}$.
For this reason, most of the works focus on the computation of bi-criteria approximations in which the persuasion constraints can be violated by a small amount. 
\citet{cheng2015mixture} present a polynomial-time bi-criteria approximation algorithm for voting scenarios.
\citet{xu2019tractability} studies public persuasion with binary action spaces and proposes a PTAS with a bi-criteria guarantee for monotone submodular sender's utility functions.
Moreover, \citet{xu2019tractability} also provides, under a non-degenerate assumption, a polynomial-time algorithm to compute an optimal signaling scheme when the number of states of nature is fixed.

\subsection{Our Results and Techniques}

%The present paper differs from~\citet{arieli2019private,babichenko2016computational,dughmi2017algorithmic} since it studies public persuasion instead of private.
%%
%Moreover, we are interested in the general setting of arbitrary state and action spaces, and arbitrary sender's utility functions.

The main result of the paper is providing a tight characterization of the complexity of computing bi-criteria approximations of optimal public signaling schemes in arbitrary persuasion problems with no inter-agent externalities.
Previous works on the same model exploit specific structures of the sender's utility functions to provide polynomial-time algorithms.
Our main result is negative, showing that restricting the space of possible sender's utility functions is a necessary condition to design polynomial-time bi-criteria approximation algorithms. 
More precisely, the following result shows that it is unlikely that there exists a bi-criteria polynomial-time approximation algorithm even in simple settings with binary action spaces.

\begin{restatable}{corollary}{mainHardness}\label{th:main_hardnessGeneral}
	Assuming ETH, there exists a constant $\epsilon^\ast$ such that, for any $\epsilon\leq\epsilon^\ast$, finding a signaling scheme that is $\epsilon$-persuasive and $\alpha$-approximate requires time $n^{\tilde\Omega(\log n)}$ for any multiplicative or additive factor $\alpha$, even for binary action spaces.
\end{restatable}
The proof of this result requires an intermediate step that is of independent interest and of general applicability. 
Specifically, we focus on a slight variation of the \textsc{Maximum Feasible Subsystem of Linear Inequalities} problem (\mssi)~\citep{cheng2015mixture}, where, given a linear system $A\,\xvec\geq 0$, $A\in[-1,1]^{\nrow\times\ncol}$, we look for the vector $\xvec\in\Delta_{\ncol}$ almost (i.e., except for an additive factor $\epsilon$) satisfying the highest number of inequalities (Definition~\ref{def:mssi}). 
This is a constrained version of the \textsc{Max FLS} problem previously studied by~\citet{amaldi1995complexity}, and it is commonly used in scheduling~\citep{daskalakis2014polynomial}, signaling, and mechanism design~\citep{cheng2015mixture}.
Assuming ETH, we prove that solving~\mssi~requires at least a quasi-polynomial number of steps via a reduction from two-provers games~\citep{aaronson2014multiple,deligkas2016inapproximability} (Section~\ref{subsec:free_games}). 

Then, we focus on a simple public persuasion problem where receivers are voters, and they have a binary action space since they must choose one between two candidates.
We prove an hardness result (Theorem~\ref{th:main_hardness}) for this setting which directly implies Corollary~\ref{th:main_hardnessGeneral}.
We show that the~\mssi~problem is deeply connected to the problem of computing ``good'' posteriors, as the choice of an optimal $\xvec$ in \mssi{} maps to the choice of an $\epsilon$-persuasive posterior.

In order to design an approximation algorithm, we resort to the assumption of $\alpha$-approximable utility functions for the sender, as previously defined by~\citet{xu2019tractability}.
An $\alpha$-approximable sender's utility function is such that it is possible to obtain in polynomial time a tie breaking for the receivers guaranteeing to the sender an $\alpha$-approximation of the optimal objective value.
%
%Following~\citet{xu2019tractability}, we say that a sender's utility function is $\alpha$-approximable if it is possible to obtain in polynomial time a tie breaking for the receivers guaranteeing to the sender an $\alpha$-approximation of the optimal objective value.
%
The request of $\alpha$-approximability is natural since otherwise even the problem of evaluating the sender's objective for a given posterior over the states of nature would not be tractable.
When the sender's utility function is $\alpha$-approximable, there is no hope for a better approximation than an $\alpha$-approximate signaling scheme. \textcolor{red}{}
The following result shows that it is possible to compute, in quasi-polynomial time, a bi-criteria approximation with a factor arbitrarily close to $\alpha$, i.e., the best factor that can be guaranteed on the objective value, and an arbitrary small loss in persuasiveness.
\begin{restatable*}{theorem}{qptas}\label{th:qptas}
	Assume $f$ is $\alpha$-approximate, there exists a $\textnormal{poly}\left(n^{\frac{\log(n/\delta)}{\epsilon^2}}\right)$ algorithm that outputs an $\alpha(1-\delta)$-approximate $\epsilon$-persuasive public signaling scheme.
\end{restatable*}
For 1-approximable functions, Theorem~\ref{th:qptas} yields a bi-criteria QPTAS. 
In the setting of~\citet{xu2019tractability} (i.e., binary action spaces and state-independent sender's utility function), our result automatically yields a QPTAS for any monotone sender's utility function.
In order to prove the result, we show that any posterior can be represented as a convex combination of $k$-uniform posteriors with only a small loss in the objective value.
By restricting our attention to the set of $k$-uniform posteriors, which has quasi-polynomial size, the problem can be solved via a linear program of quasi-polynomial size.

\section{Preliminaries}\label{sec:preliminaries}

This section describes the instantiation of the Bayesian persuasion framework which is the focus of this work (Section~\ref{subsec:model}), public signaling problems (Section~\ref{subsec:public_signaling}), the notion of bi-criteria approximation adopted (Section~\ref{subsec:bicrit}), and it presents an explanatory application to voting problems (Section~\ref{sec:application}).
For a comprehensive overview of the Bayesian persuasion framework we refer the reader to~\citet{kamenica2018bayesian,bergemann2019information} and~\citet{dughmi2017survey}.
~
\footnote{
Throughout the paper, the set $\{1,\ldots,x\}$ is denoted by $[x]$,  $\textnormal{int}(X)$ is the {\em interior} of set $X$, and $\Delta_X$ is the set of all probability distributions on $X$.
The indicator function for the event $\mathcal{E}$ is denoted by
$I[\mathcal{E}]$.
Bold case letters denote column vectors.
Moreover, we generally denote the size of a problem input by $n$. 
}

\subsection{Basic Model}\label{subsec:model}

Our model is a generalization of the fundamental special case introduced by \citet{arieli2019private}, i.e.,  multi-agent persuasion with no inter-agent externalities. 
We adopt the perspective of a {\em sender} facing a finite set of {\em receivers} $\rec\defeq[\nrec]$.
Each receiver $r$ has a finite set of $\nAct^r$ actions $\A^r\defeq\{a_i\}_{i=1}^{\nAct^r}$.
Each receiver's payoff depends only on her own action and on a (random) {\em state of nature} $\theta$, drawn from a finite set $\Theta\defeq\{\theta_i\}_{i=1}^d$ of cardinality $d$.
In particular, receiver $r$'s utility is specified by the function $u^r: \A^r \times\Theta\to [0,1]$.
Each receiver's utility does not depend on other receivers' actions for the no inter-agent externalities assumption~\cite{arieli2019private}.
We denote by $u_\theta^r(a^r)\in [0,1]$ the utility observed by receiver $r$ when the state of nature is $\theta$ and she plays $a^r$.
Let $\A\defeq \times_{r\in\rec} \A^r$. An action profile (i.e, a tuple specifying an action for each receiver) is denoted by $\avec\in \A$.
The sender's utility, when the state of nature is $\theta$, is described via the function $f_\theta: \A\to [0,1]$.
We write $f_\theta(\avec)$ to denote sender's payoff when receivers behave according to action profile $\avec$ and the state of nature is $\theta$.
As customary in the BP literature, $f_\theta$ is represented implicitly for each $\theta$ (see Equation~\ref{def:f_k_voting} for an example).

As it is customary in Bayesian persuasion, we assume $\theta$ is drawn from a common prior distribution $\mu\in\textnormal{int}(\Delta_\Theta)$, which is explicitly known to the sender and the receivers.
Moreover, the sender can publicly commit to a policy $\phi$ (i.e., a {\em signaling scheme}, see Section~\ref{subsec:public_signaling}) which maps states of nature to {\em signals} for the receivers.
A generic signal for receiver $r$ is denoted by $s^r$.
The interaction between the sender and the receivers goes as follows: 
\begin{enumerate} 
\item the sender commits to a publicly known signaling scheme $\phi$; 
\item the sender observes the realized state of nature $\theta \sim \mu$;
\item the sender draws $(s^r)_{r=1}^{\nrec}\sim \phi_\theta$ and communicates to each receiver $r$ the signal $s^r$;
\item each receiver $r$ observes $s^r$ and rationally updates her prior beliefs over $\Theta$ according to the Bayes rule. Then, each receiver selects an action maximizing her expected reward. 
\end{enumerate}
%
%(1) the sender commits to a publicly known signaling scheme $\phi$; 
%%
%(2) the sender observes the realized state of nature $\theta \sim \mu$;
%%
%(3) the sender draws $(s^r)_{r=1}^{n_r}\sim \phi_\theta$ and communicates to each receiver $r$ the signal $s^r$;
%%
%(4) each receiver $r$ observes $s^r$ and rationally updates her prior beliefs over $\Theta$. Then, each receiver selects an action maximizing her expected reward. 
%
Let $\avec$ be the tuple of receivers' choices. Each receiver $r$ observes payoff $u^r_{\theta}(a^r)$, and the sender observes payoff $f_\theta(\avec)$.

\subsection{Public Signaling Schemes}\label{subsec:public_signaling}

Each receiver $r$ has a set $\sset^r$ of available signals. 
A {\em signal profile} is a tuple $\svec=(s^r)_{r=1}^{\nrec}\in \sset$ specifying a signal for each receiver, where $\sset\defeq \times_{r\in\rec} \sset^r$. 
A {\em public signaling scheme} is a function $\phi:\Theta\to\sset$ mapping states of nature to signal profiles, with the constraint that each receiver has to receive the same signal.
With an overload of notation we write $s\in\sset$ for the public signal received by all receivers. 
The probability with which the sender selects $s$ after observing $\theta$ is denoted by $\phi_{\theta}(s)$.
Thus, it holds $\sum_{s\in\sset} \phi_\theta(s)=1$ for each $\theta\in\Theta$.
After observing $s\in \sset$, receiver $r$ performs a Bayesian update and infers a posterior belief $\pvec\in\Delta_\Theta$ over the states of nature as follows: the realized state of nature is $\theta$ with probability 
%\[
%p_\theta\defeq \frac{\mu_\theta \phi_\theta(s)}{\sum_{\theta\in\Theta}\phi_\theta(s)}.
%\]
$
p_\theta\defeq \mu_\theta \, \phi_\theta(s)/\sum_{\theta\in\Theta}\mu_\theta\,\phi_\theta(s).
$
Since the prior is common and all receivers observe the same $s$, they all perform the same Bayesian update and have the same posterior belief regarding the states of nature. 
After forming $\pvec$, each receiver solves a disjoint single-agent decision problem to find the action maximizing her expected utility.

A signaling scheme is {\em direct} when signals can be mapped to actions of the receivers, and interpreted as action recommendations.
Each receiver is sent a vector specifying a (possibly different) action for each other receiver, 
i.e., for each $r\in\rec$, $\sset^r=\A$. 
%
%We assume receivers break ties in favor of the sender.
%
%This guarantees a unique pure response for each receiver.
%
Moreover, a signaling scheme is {\em persuasive} if following the recommendations is an equilibrium of the underlying {\em Bayesian game}~\citep{bergemann2016bayes,bergemann2016information}.
A direct signaling scheme is persuasive if the sender's action recommendation belongs to the set $\argmax_{a\in\A^r} \sum_{\theta\in\Theta} p_\theta \, u^r_{\theta}(a)$.
A simple revelation-principle style argument shows that there always exists an optimal public signaling scheme which is both direct and persuasive~\citep{kamenica2011bayesian,arieli2019private}.
A signal in a direct signaling scheme can be equivalently expressed as an action profile $\avec\in \A$.
Therefore, there is an exponential number of such signals.
We write $\phi_\theta(\avec)$ to denote the probability with which the sender selects $s=\avec$ when the realized state of nature is $\theta$.
The problem of determining an optimal public signaling scheme which is direct and persuasive can be formulated with the following (exponentially sized) linear program (LP):
\begin{subequations}\label{eq:public_base}
\begin{align}\max_{\phi\geq 0} & \sum_{\theta\in\Theta,\avec\in\A}\,\mu_\theta\,\phi_\theta(\avec) \,f_\theta(\avec) \label{eq:public_base_obj}\\
\textnormal{s.t. } & \sum_{\theta\in\Theta}\mu_\theta \,\phi_\theta(\avec) \,\Big(u^r_\theta(a^{r})-u^r_\theta(a')\Big)\geq 0 & \forall r\in\rec ,\forall \avec\in\A,a'\in\A^r\label{eq:public_base_ic}\\
&\sum_{\avec\in \A}\phi_\theta(\avec)=1&\forall\theta\in\Theta
% &\phi_\theta(\avec)\geq 0 &\forall \theta\in\Theta,\forall \avec\in \A
\end{align} 
\end{subequations}
The sender's goal is computing the signaling scheme maximizing her expected utility (objective function~\ref{eq:public_base_obj}).
Constraints~\ref{eq:public_base_ic} force the public signaling scheme to be persuasive.

\subsection{Bi-criteria Approximation}\label{subsec:bicrit}

We say that a public signaling scheme is {\em $\epsilon$-persuasive} if the following holds for any $r\in\rec$, $\avec\in\A$, and $a'\in\A^r$:
\begin{equation}\label{eq:eps_persuasive}
\sum_{\theta\in\Theta}\,\mu_\theta \,\phi_\theta(\avec) \,\Big(u^r_\theta(a^{r})-u^r_\theta(a')\Big)\geq -\epsilon.
\end{equation}

Throughout the paper, we focus on the computation of approximately optimal signaling schemes.
Let \opt~be the optimal value of LP~\eqref{eq:public_base}, i.e., the best sender's expected revenue under public persuasion constraints.
Since $f_\theta$s are non-negative functions, we have that $\opt\geq0$.
When a signaling scheme yields an expected sender utility of at least $\alpha\,\opt$, with $\alpha\in(0,1]$, we say that the signaling scheme is  $\alpha${\em-approximate}.
When a signaling scheme yields an expected sender utility of at least $\opt-\alpha$, with $\alpha \in [0,1)$, we say that the scheme is {\em $\alpha$-optimal}.

Finally, we consider approximations which relax both the optimality and the persuasiveness constraints.
When a signaling scheme is both $\epsilon$-persuasive and $\alpha$-approximate (or $\alpha$-optimal), we say it is a {\em bi-criteria approximation}. 
We say that one such signaling scheme is {\em $(\alpha,\epsilon)$-persuasive}.

\subsection{An Application: Persuasion In Voting Problems}\label{sec:application}

In an election with a $k$-voting rule, candidates are elected if they receive at least $k\in [\nrec]$ votes.
In this setting, a sender (e.g., a politician or a lobbyist) may send signals to the voters on the basis of private information which is hidden from them.
After observing the sender's signal, each voter (i.e., the receivers) chooses one among the set of candidates.

In the following, we will employ instances of $k$-voting in which receivers have to choose one between two candidates.
Then, they have a binary action space with actions $a_0$ and $a_1$ corresponding to the choice of the first and the second candidate, respectively.
Each receiver $r$ has utility $u_\theta^r(a)\in[0,1]$ for each $a\in\{a_0,a_1\}$, $\theta\in\Theta$.
The sender's preferred candidate is the one corresponding to action $a_0$.
Therefore, her objective is maximizing the probability that $a_0$ receives more than $k$ votes.
Formally, the sender's utility function is such that $f_\theta=f$ for each $\theta$, and 
\begin{equation}\label{def:f_k_voting}
f(\avec)\defeq \begin{cases}
\begin{array}{ll}
1 & \textnormal{ if } |\{r\in\rec : a^r=a_0\}|\geq k\\
0 & \textnormal{ otherwise}
\end{array}
\end{cases}\textnormal{ for each } \avec\in\A.
\end{equation}
Moreover, let $W:\Delta_\Theta\to \mathbb{N}_0^+$ be a function returning, for a given posterior distribution $\pvec\in\Delta_\Theta$, the number of receivers such that $\sum_{\theta}p_\theta\,(u_\theta^r(a_0)-u_\theta^r(a_1))\geq 0$.
Analogously, $W_\epsilon(\pvec)$ is the number of receivers for which $\sum_{\theta}p_\theta\,(u_\theta^r(a_0)-u_\theta^r(a_1))\geq -\epsilon$.
In the above setting, we refer to the problem of finding an $\epsilon$-persuasive signaling scheme which is also $\alpha$-approximate (or $\alpha$-optimal) as $(\alpha,\epsilon)$-$k$-voting.
To further clarify this election scenario, we provide the following simple example, by \citet{castiglioni2019persuading}.
\begin{example}
	There are three voters $\rec=\{1,2,3\}$ who must select one between two candidates $\{a_0,a_1\}$.
	The sender (e.g., a politician or a lobbyist) observes the realized state of nature, drawn from the uniform distribution over $\Theta=\{A,B,C\}$, and exploits this information to help $a_0$ being elected.
	The state of nature describes the position of $a_0$ on a matter of particular interest to the voters.
	Moreover, all the voters  have a slightly negative opinion of candidate $a_1$, independently of the state of nature.
	Table~\ref{tab:example} describes the utility of the three voters.
	
%	\begin{table}[h]
%		\begin{minipage}{.45\linewidth}
%			\scriptsize
%			\includegraphics[scale=1.5]{figures/tab1}	
%			\caption{Payoffs from voting different candidates.}
%			\label{tab:example}
%		\end{minipage}
%		%
%		\hspace{1cm}
%		%
%		\begin{minipage}{.3\linewidth}
%			\centering
%			\includegraphics[scale=1.5]{figures/tab2}
%			\caption{Optimal signaling scheme.}
%			\label{tab:example_scheme}
%		\end{minipage}
%	\end{table}

	\renewcommand{\arraystretch}{1}\setlength{\tabcolsep}{2pt}
	\begin{table}[h]
		\begin{minipage}{.49\linewidth}
			\small
			\centering
			\begin{tabular}{rl||ccr|lccr|lccr}
				\toprule
				&&& \multicolumn{2}{c}{State $A$} &&& \multicolumn{2}{c}{State $B$} &&& \multicolumn{2}{c}{State $C$} \\
				&&& $a_0$ & $a_1$ &&& $a_0$ & $a_1$ &&& $a_0$ & $a_1$ \\
				\midrule
				\multirow{3}{*}{\rotatebox[origin=c]{90}{Voters}} &1 && $+1$ & $-1/4$ &&& $-1$ & $-1/4$ &&& $-1$ & $-1/4$\\
				&2 && $-1$ & $-1/4$ &&& $+1$ & $-1/4$ &&& $-1$ & $-1/4$\\
				&3 && $-1$ & $-1/4$ &&& $-1$ & $-1/4$ &&& $+1$ & $-1/4$\\
				\bottomrule
			\end{tabular}
			\caption{Payoffs from voting different candidates.}
			\label{tab:example}
		\end{minipage}
		\hspace{1cm}
		\begin{minipage}{.33\linewidth}
			\centering
			\small
			\begin{tabular}{lr||lcr lcr lcr}
				\toprule
				&& \multicolumn{8}{c}{Signals}&\\
				&&& \textsf{not A} &&& \textsf{not B} &&& \textsf{not C}&\\
				\midrule
				\multirow{3}{*}{\rotatebox[origin=c]{90}{States}}&$A$ && 0 &&& $1/2$ &&& $1/2$ &\\
				&$B$ && $1/2$ &&& 0 &&& $1/2$ &\\
				&$C$ && $1/2$ &&& $1/2$ &&& 0 &\\
				\bottomrule
			\end{tabular}
			\caption{Optimal signaling scheme.}
			\label{tab:example_scheme}
		\end{minipage}
		%	\caption{(a) Payoffs from voting different candidates. (b) Optimal signaling scheme.}
		%	\label{tab:big}
	\end{table}
	We consider a $k$-voting rule with $k=2$. 
	Without any form of signaling, all the voters would vote for $a_1$ because it provides an expected utility of $-1/4$, against $-1/3$.
	If the sender discloses all the information regarding the state of nature (i.e., with a fully informative signal), he would still get $0$ utility, since two out of three receivers would pick $a_1$ in each of the possible states.
	However, the sender can design a public signaling scheme guaranteeing herself utility 1 for each state of nature.
	Table~\ref{tab:example_scheme} describes one such scheme with arbitrary signals.
	Suppose the observed state is $A$, and that the signal is \textsf{not B}.
	Then, the posterior distribution over the states of nature is $(1/2, 0, 1/2)$. 
	Therefore, receiver 1 and receiver 3 would vote for $a_0$ since their expected utility would be 0 against $-1/4$.
	Similarly, for any other signal, two receivers vote for $a_0$. Then, the sender's expected payoff is $1$.
	We can recover an equivalent direct signaling scheme by sending a tuple with a candidates' suggestion for each voter.
	For example, \textsf{not A} would become $(a_1,a_0,a_0)$, and each voter would observe the recommendations given to the others.
\end{example}

\section{Technical Toolkit}

In this section, we describe some key results previously studied in the literature that we will exploit in the remainder of the paper. 
In particular, we summarize some of the results on two-prover games by~\citet{babichenko2015can} and~\citet{deligkas2016inapproximability} (Section~\ref{subsec:free_games}), and we describe a useful Theorem on error-correcting codes by~\citet{gilbert52} (Section~\ref{subsec:errorcorrectingcodes}).

\subsection{Two-Provers Games}\label{subsec:free_games}
%This section reviews some key concepts and results by~\citet{deligkas2016inapproximability}, building on previous works by~\citet{aaronson2014multiple} and~\citet{babichenko2015can}, which we exploit in the paper.

A {\em two-prover} game $\twoProvGame$ is a co-operative game played by two players (\mUno~and \mDue, resp.), and an adjudicator (verifier) called \art.
At the beginning of the game, \art~draws a pair of questions $(x,y)\in\mathcal{X}\times\mathcal{Y}$ according to a probability distribution $\mathcal{D}$ over the joint set of questions (i.e., $\mathcal{D}\in\Delta_{\mathcal{X}\times\mathcal{Y}}$).
\mUno{} (resp., \mDue) observes $x$ (resp., $y$) and chooses an answer $\ansUno$ (resp., $\ansDue$) from her finite set of answers $\ansUnoSet$ (resp., $\ansDueSet$).
Then, \art~declares the Merlins to have {\em won} with a probability equal to the value of a {\em verification function} $\ver(x,y,\ansUno,\ansDue)$.
A {\em strategy} for \mUno{} is a function $\mStrat_1:\mathcal{X}\to \ansUnoSet$ mapping each possible question to an answer.
Analogously, $\mStrat_2: \mathcal{Y}\to\ansDueSet$ is a strategy of \mDue.
Before the beginning of the game, \mUno~and \mDue~can agree on their pair of (possibly mixed) strategies $(\mStrat_1,\mStrat_2)$, but no communication is allowed during the games.
The payoff of a game $\twoProvGame$ under $(\mStratUno,\mStratDue)$ is defined as:
$
u(\twoProvGame,\mStratUno,\mStratDue)\defeq\Expec_{{(x,y)\sim\mathcal{D}}} [\ver(x,y,\mStratUno(x),\mStratDue(y))]
$
.
The {\em value} of a two-prover game $\twoProvGame$, denoted by $\omega(\twoProvGame)$, is the maximum expected payoff to the Merlins when they play optimally: $\omega(\twoProvGame)\defeq\max_{\mStratUno}\max_{\mStratDue} u(\twoProvGame,\mStratUno,\mStratDue)$.
The size of the game is $|\twoProvGame|=|\mathcal{X}\times\mathcal{Y}\times\ansUnoSet \times \ansDueSet|$.

%A two-prover game is called a {\em free game} if there is no correlation between the questions sent to \mUno~and \mDue, i.e., $\mathcal{D}$ is a uniform distribution over $\mathcal{X}\times\mathcal{Y}$.
%
A two-prover game is called a {\em free game} if $\mathcal{D}$ is a uniform distribution over $\mathcal{X}\times\mathcal{Y}$. This implies that there is no correlation between the questions sent to \mUno~and \mDue.
It is possible to build a family of free games mapping to 3SAT formula arising from Dinur's PCP theorem. 
We say that the size $n$ of a formula $\varphi$ is the number of variables plus the number of clauses in the formula. 
Moreover, SAT($\varphi$)$\in[0,1]$ is the maximum fraction of clauses that can be satisfied in $\varphi$.
With this notation, the Dinur's PCP Theorem reads as follows:
\begin{theorem}[Dinur's PCP Theorem~\citep{dinur2007pcp}]\label{th:dinur}
	Given any 3SAT instance $\varphi$ of size $n$, and a constant $\rho\in (0,\frac{1}{8})$, we can produce in polynomial time a 3SAT instance $\varphi'$ such that: 
	\begin{enumerate}
		\item the size of $\varphi'$ is $n\,\textnormal{polylog}(n)$;
		
		\item each clause of $\varphi'$ contains exactly 3 variables, and every variable is contained in at most $d=O(1)$ clauses;
		
		\item if $\textnormal{SAT}(\varphi)=1$, then $\textnormal{SAT}(\varphi')=1$;
		
		\item if $\textnormal{SAT}(\varphi)<1$, then $\textnormal{SAT}(\varphi')<1-\rho$.
	\end{enumerate}
\end{theorem}
A 3SAT formula can be seen as a bipartite graph in which the left vertices are the variables, the right vertices are the clauses, and there is an edge between a variable and a clause whenever that variable appears in that clause.
Then, a such bipartite graph has constant degree since each vertex has constant degree. This holds because each clause has at most 3 variables and each variable is contained in at most $d$ clauses. 
A useful result on bipartite graphs is the following.
\begin{lemma}[Lemma 1 of~\citet{deligkas2016inapproximability}]\label{lemma:bipartite}
	Let $(V,E)$ be a bipartite graph with $|V|=n$, and $U$ and $W$ be the two disjoints and independent sets such that $V=U\cup W$, and where each vertex has degree at most $d$.
	Suppose that $U$ and $W$ both have a constant fraction of the vertices, i.e., $|U|=cn$ and $|W|=(1-c)n$ for some $c\in[0,1]$.
	Then, we can efficiently find a partition $\{S_i\}_{i=1}^{\sqrt{n}}$ of $U$, and a partition $\{T_i\}_{i=1}^{\sqrt{n}}$ of $W$, such that each set has size at most $2\sqrt{n}$, and for all $i$ and $j$ we have $|(S_i\times T_j) \cap E|\leq 2d^2$.
\end{lemma}
Lemma~\ref{lemma:bipartite} can be used to build the following free game.
\begin{definition}[Definition 2 of~\citet{deligkas2016inapproximability}]\label{def:freegame}
	Given a 3SAT formula $\varphi$ of size $n$, we define a free game $\F_\varphi$ as follows:
	\begin{enumerate}
		\item \art~applies Theorem~\ref{th:dinur} to obtain formula $\varphi'$ of size $n\,\textnormal{polylog}(n)$;
		
		\item let $m=\sqrt{n\,\textnormal{polylog}(n)}$. \art~applies Lemma~\ref{lemma:bipartite} to partition the variables of $\varphi'$ in sets $\{S_i\}_{i=1}^{m}$, and the clauses in sets $\{T_i\}_{i=1}^{m}$;
		
		\item \art~draws an index $i$ uniformly at random from $[m]$, and independently an index $j$ uniformly at random from $[m]$.
		Then, he sends $S_i$ to \mUno~and $T_j$ to \mDue;
		
		\item \mUno~responds by choosing a truth assignment for each variable in $S_i$, and \mDue~responds by choosing a truth assignment to every variable that is involved with a clause in $T_j$;
		
		\item \art~awards the Merlins payoff 1 if and only if the following conditions are both satisfied:
			\begin{itemize}
				\item \mDue's assignment satisfies all clauses in $T_j$;
				%\item The two Merlins' assignments are {\em compatible}, i.e., for each variable $v$ appearing in $S_i$, there exists a clause of $T_j$ such that \mUno's assignment to $v$ agrees with \mDue's assignment to $v$.
				\item the two Merlins' assignments are {\em compatible}, i.e., for each variable $v$ appearing in $S_i$ and each clause in $T_j$ that contains $v$, \mUno's assignment to $v$ agrees with \mDue's assignment to $v$;
			\end{itemize}
		\art~awards payoff 0 otherwise.
	\end{enumerate}
\end{definition}

When computing Merlins' awards, the second condition is always satisfied when $S_i$ and $T_j$ share no variables.
Moreover, when \mUno's and \mDue's assignments are not compatible, we say that they are {\em in conflict}.

The following lemma shows that, if $\varphi$ is unsatisfiable, then the value of $\F_\varphi$ is bounded away from 1.
\begin{lemma}[Lemma 2 by~\citet{deligkas2016inapproximability}]\label{lemma:free_game_unsatisfiable}
	Given a 3SAT formula $\varphi$, the following holds:
	\begin{itemize}
		\item if $\varphi$ is satisfiable then $\omega(\F_\varphi)=1$;
		\item if $\varphi$ is unsatisfiable then $\omega(\F_\varphi)\leq 1-\rho/2d$.
	\end{itemize}
\end{lemma}
We prove the following original result, which follows from Lemma~\ref{lemma:free_game_unsatisfiable}.
\begin{restatable}{lemma}{lemmaMerlin}\label{lemma:merlin_conflict}
	Given a 3SAT formula $\varphi$,
	if $\varphi$ is unsatisfiable, then for each (possibly randomized) \mDue's strategy $\mStratDue$ there exists a set $S_i$ such that each \mUno's assignment to variables in $S_i$ is in conflict with \mDue's assignment with probability at least $\rho/2d$.
\end{restatable}
\begin{proof}
	Let $\omega(\F_\varphi,\mStratDue|S_i)$ be the probability with which \art~accepts Merlin's answers when \mUno~receives $S_i$, and \mDue~follows strategy $\mStratDue$. Formally:
	\[
	\omega(\F_\varphi,\mStratDue|S_i)\defeq\max_{\mStratUno}\Expec_{T_i}[\ver(S_i,T_i,\mStratUno,\mStratDue)].
	\]
	By definition of the value of a free game, we have:
	\[
	\omega(\F_\varphi)=\frac{1}{m}\max_{\mStratDue}\sum_{S_i}\omega(\F_\varphi,\mStratDue|S_i) \geq \max_{\mStratDue}\min_{S_i}\omega(\F_\varphi,\mStratDue|S_i).
	\]
	Then, by Lemma~\ref{lemma:free_game_unsatisfiable}, this results in:
	\[
	\max_{\mStratDue}\min_{S_i}\omega(\F_\varphi,\mStratDue|S_i)\leq 1-\frac{\rho}{2d},
	\]
	which proves our statement.
\end{proof}

We define \fg~as a particular problem within the class of {\em promise problems} (see, e.g.,~\citet{even1984complexity,goldreich2006promise}).
\begin{definition}[\fg]\label{def:free_game_decision}~\\
\vspace{-.4cm}
\begin{itemize}
	\item \texttt{INPUT}: a free game $\F_\varphi$ and a constant $\delta>0$.
	\item \texttt{OUTPUT}: \textsc{Yes}-instances: $\omega(\F_\varphi)=1$;  \textsc{No}-instances: $\omega(\F_\varphi)\le1-\delta$.
\end{itemize}
\end{definition}
Finally, the following lower bound will be exploited in the remainder of the paper.
We will need to assume the {\em Exponential Time Hypothesis} (ETH), which conjectures that any deterministic algorithm solving 3SAT requires $2^{\Omega(n)}$ time.
\begin{theorem}(Theorem 2 by~\citet{deligkas2016inapproximability})
	Assuming ETH, there exists a constant $\delta=\rho/2d$ such that \fg~requires time $n^{\tilde \Omega(\textnormal{log}n)}$.\footnote{$\tilde\Omega$ hides polylogarithmic factors.}
\end{theorem}

\subsection{Error-Correcting Codes}\label{subsec:errorcorrectingcodes}

%
% In this section, we introduce error-correcting codes and some of their basic properties (see~\citet{richardson2008modern} for further details).
%
A {\em message} of length $k\in\mathbb{N}_+$ is encoded as a {\em block} of length $n\in\mathbb{N}_+$, with $n\geq k$.
A {\em code} is a mapping $e:\{0,1\}^k \to \{0,1\}^n$. 
Moreover, let $\dist(e(x),e(y))$ be the {\em relative Hamming distance} between $e(x)$ and $e(y)$, which is defined as the Hamming distance weighted by $1/n$.
The {\em rate} of a code is defined as $R=\frac{k}{n}$. 
Finally, the {\em relative distance} $\dist(e)$ of a code $e$ is the maximum value $d$ such that $\dist(e(x),e(y))\ge d$ for each $x,y\in \{0,1\}^k$.

In the following, we will need an infinite sequence of codes $\mathcal{E}\defeq\{e_k:\{0,1\}^k\to\{0,1\}^n\}_{k\in\mathbb{N_{+}}}$ containing one code $e_k$ for each possible message length $k$.
The following result, due to~\citet{gilbert52}, can be used to construct an infinite sequence of codes with constant rate and distance.

\begin{theorem}[Gilbert-Varshamov Bound] \label{th:code}
	For every $k \in \mathbb{N}_+$, $0 \le d < \frac{1}{2}$ and $n \ge\frac {k}{1-\mathcal{H}_2(d)}$, there exists a code $e:\{0,1\}^k\to\{0,1\}^n$ with $\dist(e)=d$, where 
	\[
	\mathcal{H}_2(d)\defeq d \log_2\left(\frac{1}{d}\right) + (1-d)\log_2\left(\frac{1}{1-d}\right).
	\]
	Moreover, it can be computed in time $2^{O(n)}$.
\end{theorem}

\section{Maximum $\epsilon$-Feasible Subsystem of linear inequalities} 

Given a system of linear inequalities $A\,\mathbf{x}\geq 0$ with $A\in [-1,1]^{\nrow\times\ncol}$ and $\xvec\in\Delta_{\ncol}$, we study the problem of finding the largest subsystem of linear inequalities that violates the constraints of at most $\epsilon$.
As we will show in Section~\ref{sec:hardness}, this problem presents some deep connections with the problem of determining {\em good} posteriors in signaling problems. 
\begin{definition}[\mssi]\label{def:mssi}
	Given a matrix $A\in [-1,1]^{\nrow \times \ncol}$, let $\xvec^\ast\in\Delta_{\ncol}$ be the probability vector maximizing
	$
	k^\ast\defeq \sum_{i\in[\nrow]}I[w_i^\ast\geq 0],
	$
	where $\wvec^\ast\defeq A\,\xvec^\ast$.
	The problem of finding the {\em maximum $\epsilon$-feasible subsystem of linear inequalities} (%$\epsilon$-
	\mssi) amounts to finding a probability vector $\xvec\in\Delta_{\ncol}$ such that, by letting $\wvec= A\,\xvec$, it holds: $\sum_{i\in[\nrow]}I[w_i\geq -\epsilon]\geq k^\ast$.
\end{definition}

This problem is previously studied by~\citet{cheng2015mixture}.
They design a bi-criteria PTAS for the \mssi{} problem guaranteeing the satisfaction of at least $k^\ast-\epsilon\,\nrow$ inequalities.
Initially, we show that \mssi{} can be solved in $n^{O(\log n)}$ steps.
We introduce the following auxiliary definition.
\begin{definition}[{\em $k$-uniform distribution}]\label{def:k_distribution}
%	A probability distribution $\xvec\in \Delta_X$ is {\em $k$-uniform} if, for each $i\in[|X|]$, $x_i$ is a multiple of $1/k$.
%
A probability distribution $\xvec\in\Delta_X$ is {\em $k$-uniform} if and only if it is the average of a multiset of $k$ basis vectors in $|X|$-dimensional space.
\end{definition}
Equivalently, each entry $x_i$ of a $k$-uniform distribution has to be a multiple of $1/k$.
Then, the following result holds.
\begin{theorem}\label{th:mssi_upper}
	\mssi~can be solved in $n^{O(\log n)}$ steps.
\end{theorem}
\begin{proof}
	Denote by $\xvec^\ast$ the optimal solution of \mssi.
	Let $\tilde\xvec\in\Delta_{\ncol}$ be the empirical distribution of $k$ i.i.d. samples drawn from probability distribution $\xvec^\ast$.
	Moreover, let $\wvec^\ast\defeq A\,\xvec^\ast$ and $\tilde\wvec\defeq A\,\tilde\xvec$.
	By Hoeffding's inequality we have $\pr(w_i^\ast-\tilde w_i\geq\epsilon)\leq e^{-2k\epsilon^2}$ for each $i\in[\nrow]$.
	Then, by the union bound, we get $\pr(\exists i \textnormal{ s.t. } w_i^\ast-\tilde w_i\geq\epsilon)\leq \nrow e^{-2k\epsilon^2}$.
	Finally, we can write $\pr(w_i^\ast-\tilde w_i\leq\epsilon ~\forall i\in[\nrow])\geq 1-\nrow e^{-2k\epsilon^2}$.
	Thus, setting $k=\log \nrow/\epsilon^2$ ensures the existence of a vector $\tilde x$ guaranteeing that, if $w_i^\ast\ge0$, then $\tilde w_i\ge -\epsilon$.
	Since $\tilde x$ is $k$-uniform by construction, we can find it by enumerating over all the $O((\ncol)^k)$ $k$-uniform probability vectors where $k=\log \nrow/\epsilon^2$. Trivially, this task can be performed in $n^{\log \nrow/\epsilon^2}$ steps and, therefore, in $n^{O(\log n)}$ steps.
\end{proof}

We show that \mssi{} requires at least $n^{\tilde\Omega(\log n)}$ steps, thus closing the gap with the upper bound stated by Theorem~\ref{th:mssi_upper} except for polylogarithmic factors of $\log n$ in the denominator of the exponent.
\begin{theorem} \label{th:mssi_lower}
	Assuming ETH, there exists a constant $\epsilon>0$ such that solving \mssi~requires time $n^{\tilde\Omega(\log n)}$.
\end{theorem}
\begin{proof}
	\MyParagraph{Overview.} We provide a polynomial-time reduction from \fg{} (Def.~\ref{def:freegame}) to \mssi, where $\epsilon=\frac{\delta}{26}=\frac{\rho}{52 d}$ (see Section~\ref{subsec:free_games} for the definition of parameters $\delta, \rho, d$). 
	We show that, given a free game $\F_\varphi$ instance, it is possible to build a matrix $A$ s.t., for a certain value $k^\ast$, the following holds: 
	(i) if $\omega(\F_\varphi)=1$, then there exists a vector $\xvec$ s.t. 
	\begin{equation}\label{eq:cond1_mssi_reduction}
	\sum_{i\in[\nrow]} I[w_i\geq 0]=k^\ast,
	\end{equation}
	where $\wvec=A\,\xvec$; 
	(ii) if $\omega(\F_\varphi)\leq 1-\delta$, then all vectors $\xvec$ are s.t.
	\begin{equation}\label{eq:cond2_mssi_reduction}
	\sum_{i\in[\nrow]} I[w_i\geq -\epsilon]<k^\ast.
	\end{equation}

	\MyParagraph{Construction.}
	In the free game $\F_\varphi$, \art~sends a set of variables $S_i$ to \mUno~and a set of clauses $T_j$ to \mDue, where $i,j\in[m]$, $m=\sqrt{n\,\textnormal{polylog}(n)}$.
	Then, \mUno's (resp., \mDue's) answer is denoted by $\ansUno\in\ansUnoSet$ (resp., $\ansDue\in\ansDueSet$).
	The system of linear inequalities used in the reduction has a vector of variables $\xvec$ structured as follows.
	\begin{enumerate}
	\item {\em Variables corresponding to \mDue's answers.} There is a variable $x_{T_j,\ansDue}$ for each $j\in [m]$ and, due to Lemma~\ref{lemma:bipartite} and the assumption $|T_j|=2m$, it holds $\ansDue\in\ansDueSet=\{0,1\}^{6m}$  (if $|T_j|< 2m$, we extend $\ansDue$ with extra bits).
	\item {\em Variables corresponding to \mUno's answers.} 
	We need to introduce some further machinery to augment the dimensionality of $\ansUnoSet$ via a viable mapping.
	Let $e:\{0,1\}^{2m}\to \{0,1\}^{8m}$ be the code stated in Theorem~\ref{th:code} with rate $1/4$ and relative distance $\dist (e)\ge1/5$.
	We can safely assume that $\ansUno \in\ansUnoSet=\{0,1\}^{2m}$ (if $|S_i|< 2m$, we extend $\ansUno$ with extra bits).
	Then, $e(\ansUno)$ is the $8m$-dimensional encoding of answer $\ansUno$ via code $e$.
	Let $e(\ansUno)_j$ be the $j$-th bit of vector $e(\ansUno)$. 
	We have a variable $x_{i,\ell}$ for each index $i\in[8m]$ and $\ell\defeq\{\ell_j\}_{j\in [m]}\in\{0,1\}^{m}$.
	These variables can be interpreted as follows.
	Suppose to have an answer encoding for each of the possible set $S_j$.
	There are $m$ such encodings, each of them having $8m$ bits.
	Then, it holds  $x_{i,\ell}>0$ if the $i$-th bit of the encoding corresponding to $S_j$ is $\ell_j$. 
	\end{enumerate}
There is a total of $m\,2^m\,(2^{5m}+8)$ variables. 
Matrix $A$ has a number of columns equal to the number of variables. 
We denote with $A_{\cdot,(T_j,\ansDue)}$ the column with the same index of variable $x_{T_j,\ansDue}$.
Analogously, $A_{\cdot,(i,\ell)}$ is the column corresponding to variable $x_{i,\ell}$.
Rows are grouped in four types, denoted by $\{\type_i\}_{i=1}^4$.
We write $A_{\type_i,\cdot}$ when referring to an entry of {\em any} row of type $\type_i$.
Further arguments may be added as a subscript to identify specific entries of $A$.
Rows are structured as follows.
\begin{enumerate}
	\item {\em Rows of type $\type_1$}: there are $q$ (the value of $q$ is specified later in the proof) rows of type $\type_1$ s.t.~$A_{\type_1,(T_j,\ansDue)}=1$ for each $j\in [m],\ansDue\in\ansDueSet$, and $A_{\type_1,\cdot}=-1$ otherwise.
	
	\item {\em Rows of type $\type_2$}: there are $q$ rows for each subset $\mathcal{T}\subseteq \{T_j\}_{j\in[m]}$ with cardinality $m/2$ (i.e., there is a total of $q\binom{m}{m/2}$ rows of type $\type_2$). 
	%
%	\red{A generic entry of a type $\type_2$ row associated to a subset $\mathcal{T}$ is denoted by $A_{(\type_2,\mathcal{T}),\cdot}$.LO USIAMO?}
	%
	Then, the following holds for each $\mathcal{T}$: 
	\begin{equation*}
		A_{(\type_2,\mathcal{T}),(T_j,\ansDue)}=\begin{cases}\begin{array}{ll}
		-1 & \textnormal{ if } T_j\in\mathcal{T},\ansDue\in\ansDueSet \\
		1 & \textnormal{ if } T_j\notin\mathcal{T},\ansDue\in\ansDueSet 
		\end{array}
		\end{cases}
		\textnormal{ and }\quad
		A_{(\type_2,\mathcal{T}),(i,\ell)}=0 \textnormal{ for each } i\in[8m], \ell\in\{0,1\}^{m}.
	\end{equation*}

	\item {\em Rows of type $\type_3$}: there are $q$ rows of type $\type_3$ for each subset of $4m$ indices $\mathcal{I}$ drawn from $[8m]$, for a total of $q\binom{8m}{4m}$ $\type_3$ rows.
	For each subset of indices $\mathcal{I}$ we have:
	\begin{equation*}
	A_{(\type_3,\mathcal{I}),(T_j,\ansDue)}=0 \textnormal{ for each } T_j,\ansDue\quad  \textnormal{ and }\quad 
	A_{(\type_3,\mathcal{I}),(i,\ell)}=
	\begin{cases}\begin{array}{ll}
	-1 & \textnormal{ if } i\in\mathcal{I},\ell\in\{0,1\}^m\\
	1 & \textnormal{ if } i\notin\mathcal{I},\ell\in\{0,1\}^m
	\end{array}.
	\end{cases}
	\end{equation*}
	
	\item {\em Rows of type $\type_4$}: there is a row of type $\type_4$ for each $S_i$ and $\ansUno$. 
	Each of these rows is s.t.:
	\begin{equation*}
	A_{(\type_4,S_i,\ansUno),(T_j,\ansDue)} = \begin{cases}
	\begin{array}{ll}
	-1/2 & \textnormal{if } \ver(S_i,T_j,\ansUno,\ansDue)=1\\
	-1 & \textnormal{otherwise}
	\end{array}
	\end{cases}
	\hspace{-.3cm}
	\textnormal{and }\quad
		A_{(\type_4,S_i,\ansUno),(j,\ell)} = \begin{cases}
	\begin{array}{ll}
	1/2 & \textnormal{if } e(\ansUno)_j=\ell_i \\
	-1 & \textnormal{otherwise}
	\end{array}
	\end{cases}\hspace{-.4cm}.
	\end{equation*}
\end{enumerate}

Finally, we set $k^\ast= \left(1+\binom{m}{m/2}+\binom{8m}{4m}\right)q + m$ and $q\gg m$ (e.g., $q=2^{10m}$).
We say that row $i$ satisfies \mssi~condition for a certain $\xvec$ if $w_i\geq -\epsilon$, where $\wvec=A\,\xvec$ (in the following, we will also consider $w_i\geq 0$ as an alternative condition).
We require at least $k^\ast$ rows to satisfy the \mssi{} condition.
Then, all rows of types $\type_1$, $\type_2$, $\type_3$ and at least $m$ rows of type $\type_4$ must be s.t.~$w_i$ satisfies the condition.

\MyParagraph{Completeness.}
Given a satisfiable assignment of variables $\ass$ to $\varphi$, we build vector $\xvec$ as follows.
Let $\ass_{T_j}$ be the partial assignment obtained by restricting $\ass$ to the variables in the clauses of $T_j$ (if $|T_j|<2m$ we pad $\ass_{T_j}$ with bits 0 until $\ass_{T_j}$ has length $6m$).
Then, we set $x_{T_j,\ass_{T_j}}=1/2m$.
Moreover, for each $i\in[8m]$ and $\ell^i=(e(\ass_{S_1})_i,\ldots,e(\ass_{S_m})_i)$, we set $x_{i,\ell^i}=1/16m$.
We show that $\xvec$ is s.t.~there are at least $k^\ast$ rows $i$ with $w_i\geq 0$ (Condition~\eqref{eq:cond1_mssi_reduction}).
First, each row $i$ of type $\type_1$ is s.t.~$w_i=0$ since $\sum_{T_j,\ansDue} x_{T_j,\ansDue}=\sum_{i,\ell} x_{i,\ell}=1/2$. 
For each $T_j$, $\sum_{\ansDue}x_{T_j,\ansDue}=1/2m$.
Then, for each subset $\mathcal{T}$ of $\{T_j\}_{j\in[m]}$, we have $\sum_{\ansDue,T_j\in\mathcal{T}} x_{T_j,\ansDue}=1/4$.
This implies that each row $i$ of type $\type_2$ is s.t.~$w_i=0$.
A similar argument holds for rows of type $\type_3$.
Finally, we show that for each $S_i$ there is at least a row $i$ of type $\type_4$ s.t.~$w_i\geq 0$.
Take the row corresponding to $(S_i,\ass_{S_i})$. 
For each $x_{b,\ell}>0$ where $b\in [8m]$ and $\ell\in\{0,1\}^m$, it holds $e(\ass_{S_i})_b=\ell_i$.
Then, there are $8m$ columns played with probability $1/16m$ with value $1/2$, i.e., $\sum_{b,\ell} A_{(\type_4,S_i,\zeta_{S_i}),(b,\ell)}x_{b,\ell}=1/4$.
Moreover, for each $(T_j,\zeta_{T_j})$, it holds $\ver(S_i,T_j,\ass_{S_i},\ass_{T_j})=1$.
Then, $\sum_{T_j,\ansDue}A_{(\type_4,S_i,\ass_{S_i}),(T_j,\ass_{T_j})}x_{T_j,\ansDue}=-1/4$. 
This concludes the completeness section.

\MyParagraph{Soundness.}
We show that, if $\omega(\F_\varphi)\leq 1-\delta$, there is not any probability distribution $\xvec$ s.t. 
\begin{equation}\label{eq:condition_soundness}
\sum_{i\in \nrow} I[w_i\geq -\epsilon]\geq k^\ast,
\end{equation}
with $\wvec=A\,\xvec$.
Assume, by contradiction, that one such vector $\xvec$ exists.
For the sake of clarity, we summarize the structure of the proof.
(i) We show that the probability assigned by $\xvec$ to columns of type $(T_j,\ansDue)$ has to be {\em close} to $1/2$, and the same has to hold for columns of type $(i,\ell)$.
(ii) We show that $\xvec$ has to distribute probability {\em almost} uniformly among $T_j$s and indices $i$ (resp., Lemma~\ref{lemma:almost_uniform_1} and Lemma~\ref{lemma:almost_uniform_2} below).
Intuitively, this resembles the fact that, in $\F_\varphi$, \art~draws questions $T_j$ according to a uniform probability distribution.
(iii) For each $S_i$, there is at most one row $(\type_4,S_i,\ansUno)$ s.t.~$w_{(\type_4,S_i,\ansUno)}\geq -\epsilon$ (Lemma~\ref{lemma:at_most_one_row}).
This implies, together with the hypothesis, that there exists exactly one such row for each $S_i$.
(iv) Finally, we show that the above construction leads to a contradiction with Lemma~\ref{lemma:merlin_conflict} for a suitable free game.

Before providing the details of the four above steps, we introduce the following result, due to~\citet{babichenko2015can}.
\begin{lemma}[Essentially Lemma 2 of~\citet{babichenko2015can}]\label{lemma:deviation_from_uniform}
	Let $\mathbf{v}\in\Delta^n$ be a probability vector, and $\mathbf{u}$ be the $n$-dimensional uniform probability vector.
	If $||\mathbf{v}-\mathbf{u}||>c$, then there exists a subset of indices $\mathcal{I}\subseteq [n]$ such that $|\mathcal{I}|=n/2$ and $\sum_{i\in \mathcal{I}}\mathbf{v}_i>\frac{1}{2}+\frac{c}{4}$. 
\end{lemma}
Then, 

\begin{enumerate}[label=(\roman*),leftmargin=.7cm]
	\item Equation~\ref{eq:condition_soundness} requires all rows $i$ of type $\type_1$, $\type_2$, $\type_3$ to be s.t.~$w_i\geq -\epsilon$.
	This implies that, for rows of type $\type_1$, it holds $\sum_{T_j,\ansDue} x_{T_j,\ansDue}\geq 1/2(1-\epsilon)$.
	Indeed, if, by contradiction, this inequality does not hold, each row $i$ of type $\type_1$ would be s.t.~$w_i<1/2-\epsilon/2-(1/2+\epsilon/2)=-\epsilon$, thus violating Equation~\ref{eq:condition_soundness}.
	Moreover, Equation~\ref{eq:condition_soundness} implies that at least a row $(\type_4,S_i,\ansUno)$ has $w_{(\type_4,S_i,\ansUno)}\geq -\epsilon$.
	Therefore, it holds $\sum_{i,\ell}x_{i,\ell}\geq 1/2-\epsilon$.
	Indeed, if, by contradiction, this condition did not hold, all rows of type $\type_4$ would have $w_i<1/2(1/2-\epsilon)-1/2(1/2+\epsilon)=-\epsilon$.
	
	\item Let $\vvec_1\in\Delta_m$ be the probability vector defined as $v_{1,j}\defeq \frac{\sum_{\ansDue} x_{T_j,\ansDue}}{\sum_{j,\ansDue}x_{T_j,\ansDue}}$, and $\tilde\vvec$ be a generic uniform probability vector of suitable dimension.
	The following result shows that the element-wise difference between $\vvec_1$ and $\tilde\vvec$ has to be bounded if Equation~\ref{eq:condition_soundness} has to be satisfied.
	\begin{restatable}{lemma}{almostUniformOne}\label{lemma:almost_uniform_1}
		If $||\vvec_1-\tilde\vvec||_1>16\epsilon$, there exists a row $i$ of type $\type_2$ such that $w_i<-\epsilon$.
	\end{restatable}
	\begin{proof}
		Lemma~\ref{lemma:deviation_from_uniform} implies that, if $||\vvec_1-\tilde\vvec||_1>16\epsilon$, there exists a subset $\mathcal{T}\subseteq \{T_j\}_{j\in [m]}$ such that
		$\sum_{T_j\in\mathcal{T}}\sum_{\ansDue}x_{T_j,\ansDue}>(1/2+4\epsilon)\sum_{j,\ansDue}x_{T_j,\ansDue} > 1/4+\epsilon$.
		It follows that $\sum_{T_j\notin\mathcal{T}} \sum_{\ansDue}x_{T_j,\ansDue}<1/2+\epsilon-1/4-\epsilon=1/4$, which implies that row $(\type_2,\mathcal{T})$ is s.t. $w_{\type_2,\mathcal{T}}<-1/4-\epsilon+1/4<-\epsilon$.	
	\end{proof}

	Let $\vvec_2\in\Delta_{[8m]}$ be the probability vector defined as $v_{2,i}\defeq \frac{\sum_{\ell}x_{i,\ell}}{\sum_{i,\ell}x_{i,\ell}}$, and $\tilde\vvec$ be a suitable uniform probability vector.
	The following holds.
	
	\begin{restatable}{lemma}{almostUniformTwo}\label{lemma:almost_uniform_2}
		If $||\vvec_2-\tilde\vvec||_1>16\epsilon$, there exists a row $i$ of type $\type_3$ such that $w_i<-\epsilon$.
	\end{restatable}
	\begin{proof}
		Lemma~\ref{lemma:deviation_from_uniform} implies that, if $||\vvec_2-\tilde\vvec||_1>16\epsilon$, there exists a set $\mathcal{I}\subseteq [8m]$ such that $\sum_{i\in\mathcal{I}}\sum_{\ell}x_{i,\ell}>(1/2+4\epsilon)\sum_{i,\ell}x_{i,\ell}>1/4+\epsilon$.
		Then, $\sum_{i\notin\mathcal{I}}\sum_{\ell}x_{i,\ell}<1/2 + \epsilon/2 -1/4 -\epsilon=1/4-\epsilon/2$. 
		It follows that there exists a row $(\type_3,\mathcal{I})$ such that $w_{\type_3,\mathcal{I}}<-1/4-\epsilon+1/4-\epsilon/2<-\epsilon$.
	\end{proof}

	In order to satisfy Equation~\ref{eq:condition_soundness}, all rows $i$ of type $\type_2$ and $\type_3$ have to be s.t.~$w_i\geq -\epsilon$.
	Then, by Lemmas~\ref{lemma:almost_uniform_1} and~\ref{lemma:almost_uniform_2}, it has to hold that $||\vvec_1-\tilde\vvec||_1\leq 16\epsilon$ and $||\vvec_2-\tilde\vvec||_1\leq 16\epsilon$. 
	
	\item We show that, for each $S_i$, there exists at most one row $(\type_4, S_i,\ansUno)$ for which $w_{(\type_4,S_i,\ansUno)}\geq -\epsilon$.
	\begin{restatable}{lemma}{atMostOne}\label{lemma:at_most_one_row}
		For each $S_i$, $i\in [m]$, there exists at most one row $(\type_4, S_i,\ansUno)$ s.t.~$w_{(\type_4,S_i,\ansUno)}\geq -\epsilon$.
	\end{restatable}
	\begin{proof}
		Let $\auxFunc(\xvec,\ansUno)\defeq \sum_{j:\ell_i=e(\ansUno)_j} x_{j,\ell}$.
		Assume, by contradiction, that for a given $S_i$ there exist two assignments $\ansUno'$ and $\ansUno''$ such that $w_{(\type_4,S_i,\ansUno)}\geq -\epsilon$ for each $\ansUno\in\{\ansUno',\ansUno''\}$. 
		Then, $\auxFunc(\xvec,\ansUno)\geq 1/2 -\epsilon$, for each $\ansUno\in\{\ansUno',\ansUno''\}$.
		Otherwise we would get $w_{(\type_4,S_i,\ansUno)}<1/2(1/2-\epsilon)-1/2(1/2+\epsilon)=-\epsilon$ for at least one $\ansUno\in\{\ansUno',\ansUno''\}$.
		Let $\xvec'$ be the vector such that $x'_{i,\ell}\defeq \frac{x_{i,\ell}}{\sum_{i,\ell}x_{i,\ell}}$.
		Then, $\auxFunc(\xvec',\ansUno)\geq \frac{1/2-\epsilon}{1/2+\epsilon}\geq 1-4\epsilon$, for $\ansUno\in\{\ansUno',\ansUno''\}$.
		By Lemma~\ref{lemma:deviation_from_uniform} and~\ref{lemma:almost_uniform_2}, we have that $||\vvec_2-\tilde \vvec||_1\leq 16\epsilon$. 
		Therefore, we can obtain a uniform vector $\tilde \xvec$ by moving at most $16\epsilon$ probability from $\xvec'$.
		This results in a decrease of $\auxFunc$ of at most $16\epsilon$, that is $\auxFunc(\tilde\xvec,\ansUno)\geq 1-20\epsilon$ for each $\ansUno\in\{\ansUno',\ansUno''\}$.
		
		By construction $\dist(e)=1/5$, which implies $\dist(e(\ansUno'),e(\ansUno''))\geq 1/5$.
		Then, there exists a set of indices $\mathcal{I}$, with $|\mathcal{I}|\geq 8m/5$, such that $e(\ansUno')_j\neq e(\ansUno'')_j$ for each $j\in\mathcal{I}$.
		Therefore, $\auxFunc(\tilde\xvec,\ansUno')+\auxFunc(\tilde\xvec,\ansUno'')\leq \sum_{j\in\mathcal{I}}1/8m+\sum_{j\notin\mathcal{I}}2/8m\leq 2-1/5$.
		This leads to a contradiction with $\auxFunc(\tilde\xvec,\ansUno')+\auxFunc(\tilde\xvec,\ansUno'')\geq 2-40\epsilon$.
	\end{proof}

	Then, there are at least $m$ rows $(\type_4,S_i,\ansUno)$ s.t.~$w_{(\type_4,S_i,\ansUno)}\geq -\epsilon$ and, by Lemma~\ref{lemma:at_most_one_row}, we get that there exists exactly one such row for each $S_i$, $i\in [m]$.
	Therefore, for each $S_i$, there exists $\ansUno^i\in\ansUnoSet$ s.t.~$\sum_{(T_j,\ansDue):\ver(S_i,T_j,\ansUno^i,\ansDue)=1}x_{(T_j,\ansDue)}\geq 1/2-4\epsilon$.
	If this condition did not hold, by Step (i), we would obtain  
	$w_{\type_4,S_i,\ansUno^i}<-1/2(1/2-4\epsilon)-7/2\epsilon+1/2(1/2+\epsilon/2)=-\epsilon$.
	
	\item Finally, let $\F_\varphi^\ast$ be a free game in which \art~(i.e., the verifier) chooses question $T_j$ with probability $v_{1,j}$ as defined in Step (ii), and \mDue~(i.e., the second prover) answers $\ansDue$ with probability $x_{T_j,\ansDue}/v_{1,j}$.
	In this setting (i.e., $\F_\varphi^\ast$), given question $S_i$ to \mUno, the two provers will provide compatible answers with probability $\mathbb{P}(\ver^\ast(S_i,T_j,\ansUno^i,\ansDue)=1\mid S_i)=\frac{1/2-4\epsilon}{\sum_{j,\ansDue}x_{T_j,\ansDue}}\geq \frac{1/2-4\epsilon}{1/2+\epsilon}\geq 1-10\epsilon$, where the first inequality holds for the condition at Step (i).
	In a canonical  (i.e., as in Definition~\ref{def:freegame}) free game $\F_\varphi$, \art~picks questions according to a uniform probability distribution.
	The main difference between $\F_\varphi$ and $\F_\varphi^\ast$ is that, in the latter, \art~draws questions for \mDue~from $\vvec_1$.
	However, we know that differences between $\vvec_1$ and a uniform probability vector must be limited.
	Specifically, by Lemma~\ref{lemma:almost_uniform_1}, we have $||\vvec_1-\tilde\vvec||_1\leq 16\epsilon$.
	Then, if \mUno~and \mDue~applied in $\F_\varphi$ the strategies we described for $\F_\varphi^\ast$, their answers would be compatible with probability at least $\mathbb{P}(\ver(S_i,T_j,\ansUno^i,\ansDue)=1\mid S_i)\geq 1-26\epsilon$, for each $S_i$. 
	Finally, by picking $\epsilon=\rho/52d$, we reach a contradiction with Lemma~\ref{lemma:merlin_conflict}. 
	This concludes the proof.
\end{enumerate}
\end{proof}

\section{The Hardness of $(\alpha,\epsilon)$-persuasion}\label{sec:hardness}

%We prove that even allowing an -violation of the incentive constraints, a signaling
%scheme that approximates the optimal one cannot be computed in polynomial time. In
%particular, we focus on a very simple function used in voting settings and we prove
%that the optimal signaling scheme cannot be approximated to within any factor in less
%that n Ω̃(log n) .

We show that a public signaling scheme approximating the value of the optimal one cannot be computed in polynomial time even if we allow it to be $\epsilon$-persuasive (see Equation~\ref{eq:eps_persuasive}).
Specifically, computing an $(\alpha,\epsilon)$-persuasive signaling scheme requires at least $n^{\tilde\Omega(\log n)}$, where the dimension of the instance is $n=O(\bar n \, d)$.
We prove this result for the specific case of the $k$-voting problem, as introduced in Section~\ref{sec:application}.
Besides its practical applicability, this problem is particularly instructive in highlighting the strong connection between the problem of finding suitable posteriors and the \mssi~problem, as discussed in the following lemma.
\begin{lemma}\label{lemma:posterior_mssi}
	Given a $k$-voting instance,
	the problem of finding a posterior $\pvec\in\Delta_\Theta$ such that $W_\epsilon(\pvec)\geq 0$ is equivalent to finding an $\epsilon$-feasible subsystem of $k$  linear inequalities over the simplex when $A\in[-1,1]^{\bar n\times d}$ is such that: 
	\begin{equation}\label{eq:matrix_equivalence}
		A_{r,\theta}=u^r_\theta(a_0)-u_\theta^r(a_1) \quad\textnormal{ for each }\quad r\in\rec,\theta\in\Theta.
	\end{equation}
\end{lemma}
\begin{proof}
	By setting $\xvec=\pvec$, it directly follows that $\sum_{i\in[\bar n]}I[A_i\xvec\geq -\epsilon]\geq k$ iff $W_\epsilon(\pvec)\geq k$.
\end{proof}
The above lemma shows that deciding if there exists a posterior $\pvec$ such that $W(\pvec)\geq k$ or if {\em all} the posteriors have $W_\epsilon(\pvec)<k$ (i.e., deciding if the utility of the sender can be greater than zero) is as hard as solving \mssi.
%
%We already showed that deciding if there exists a posterior $\mu$ such that $W(\mu)\geq k$ or if {\em all} the posteriors have $W_\epsilon(\mu)<k$ (i.e., deciding if the utility of the sender can be greater than zero) is as hard as solving \mssi.
%
More precisely, if the \mssi~instance does not admit any solution, then there does not exist a posterior guaranteeing the sender strictly positive winning probability.
On the other hand, if the \mssi~instance admits a solution, there exists a signaling scheme where at least one of the induced posteriors guarantees the sender wining probability $>0$.
However, the above connection between the \mssi{} problem and the $k$-voting problem is not sufficient to prove the inapproximability of the $k$-voting problem, as the probability whereby this posterior is reached may be arbitrarily small.
%
%Moreover, this gap would not allow us to provide any additive inapproximability consideration.
%
%***Finally, we can prove our main result.***

Luckily enough, the next theorem shows that it is possible to strengthen the inapproximability result by constructing an instance in which, when 3SAT is satisfiable, there is a signaling scheme such that all the induced posteriors satisfy $W(\pvec)\geq k$ (i.e., the sender wins with probability 1).

\begin{theorem}\label{th:main_hardness}
	Given a $k$-voting instance and assuming ETH, there exists a constant $\epsilon^\ast$ such that, for any $\epsilon\leq\epsilon^\ast$, finding an $(\alpha,\epsilon)$-persuasive signaling scheme requires $n^{\tilde\Omega(\log n)}$ steps for any multiplicative or additive factor $\alpha$.
\end{theorem}
\begin{proof}
	\MyParagraph{Overview.} By following the proof of Theorem~\ref{th:mssi_upper}, we provide a polynomial-time reduction from \fg~to the problem of finding an $\epsilon$-persuasive signaling scheme in $k$-voting, with $\epsilon=\delta/780=\rho/1560d$.
	Specifically, if $\omega(\F_\varphi)=1$, there exists a signaling scheme guaranteeing the sender an expected value of 1.
	Otherwise, if $\omega(\F_\varphi)\leq 1-\delta$, then all posteriors are such that $W_\epsilon(\pvec)<k$ (i.e., the sender cannot obtain more than 0).
	
	\MyParagraph{Construction.} The $k$-voting instance has the following possible states of nature.
	\begin{enumerate}
		\item $\theta_{(T_j,\ansDue)}$ for each set of clauses $T_j$, $j\in[m]$, and answer $\ansDue\in\ansDueSet=\{0,1\}^{6m}$.
		 Let $e:\{0,1\}^{2m}\to \{0,1\}^{8m}$ be an encoding function with $R=1/4$ and $\dist(e)\ge1/5$ (as in the proof of Theorem~\ref{th:mssi_upper}).
		We have a state $\theta_{(i,\ell)}$ for each $i\in [8m]$, and $\ell=(\ell_1,\ldots,\ell_m)\in\{0,1\}^m$.
		
		\item There is a state $\theta_{\dvec}$ for each $\dvec\in\{0,1\}^{7m}$.
		It is useful to see vector $\dvec$ as the union of the subvector $\dvec_S\in\{0,1\}^m$ and the subvector $\dvec_T\in\{0,1\}^{6m}$.
	\end{enumerate}
	The shared prior $\mu$ is such that: $\mu_{\theta_{(T_j,\ansDue)}}=\frac{1}{m 2^{2+6m}}$ for each $\theta_{(T_j,\ansDue)}$, $\mu_{\theta_{(i,\ell)}}=\frac{1}{m2^{5+m}}$ for each $\theta_{(i,\ell)}$, and $\mu_{\theta_\dvec}=\frac{1}{2^{1+7m}}$ for each $\theta_\dvec$.
	To simplify the notation, in the remaining of the proof let $u_\theta^r\defeq u_\theta^r(a_0)-u_\theta^r(a_1)$.
	The $k$-voting instance comprises the following receivers.
	\begin{enumerate}
		\item {\em Receivers of type $\type_1$}: there are $q$ (the value of $q$ is specified later in the proof) receivers of type $\type_1$, which are such that $u^{\type_1}_{\theta_{(T_j,\ansDue)}}=1$ for each $(T_j,\ansDue)$, and $-1/3$ otherwise.
		
		\item {\em Receivers of type $\type_2$:} there are $q$ receivers of type $\type_2$ such that $u^{\type_2}_{\theta_{(i,\ell)}}=1$ for each $(i,\ell)$, and $-1/3$ otherwise.
		
		\item {\em Receivers of type $\type_3$:} there are $q$ receivers of type $\type_3$ for each subset $\tcal\subseteq \{T_j\}_{j\in[m]}$ of cardinality $m/2$.
		Each receiver corresponding to the subset $\tcal$ is such that:
		\begin{equation*}
		u^{(\type_3,\tcal)}_{\theta_{(T_j,\ansDue)}}=
		\begin{cases}\begin{array}{ll}
		-1 & \textnormal{ if } T_j\in \tcal,\ansDue\in\ansDueSet\\
		1 & \textnormal{ if } T_j \notin \tcal, \ansDue\in\ansDueSet
		\end{array}
		\end{cases}
		\textnormal{ and }
		u_{\cdot}^{(\type_3,\tcal)}=0 \textnormal{ otherwise.}
		\end{equation*}
		
		\item {\em Receivers of type $\type_4$:}
		we have $q$ receivers ot type $\type_4$ for each subset $\ical$ of $4m$ indices selected from $[8m]$.
		Each receiver corresponding to subset $\ical$ is such that:
		\begin{equation*}
		u_{\theta_{(i,\ell)}}^{(\type_4,\ical)} = 
			\begin{cases}\begin{array}{ll}
			-1 & \textnormal{ if }i\in\ical,\ell\in\{0,1\}^m\\
			1 & \textnormal{ if } i\notin\ical,\ell\in\{0,1\}^m 
			\end{array}
			\end{cases}
			\textnormal{ and }
			u_{\cdot}^{(\type_4,\ical)}=0 \textnormal{ otherwise.}
		\end{equation*}
		
		\item {\em Receivers of type $\type_5$:}
		there is a receiver of type $\type_5$ for each $S_i$, $\ansUno\in\ansUnoSet$ and $\dvec\in\{0,1\}^{7m}$.
		Let $\oplus$ be the XOR operator.
		Then, for each receiver of type $\type_5$ the following holds:
		\begin{equation*}
		u_{\theta}^{(\type_5,S_i,\ansUno,\dvec)} = 
		\begin{cases}\begin{array}{ll}
		 -1/2 & \textnormal{ if }\theta=\theta_{(T_j,\ansDue)} \textnormal{ and }\ver(S_i,T_j,\ansUno,\ansDue\oplus \dvec_{T})=1\\
		 -1/2 & \textnormal{ if }\theta = \theta_{(i',\ell)}\textnormal{ and } e(\ansUno)_{i'}=[\ell\oplus\dvec_S]_{i}\\
		 1/2 & \textnormal{ if } \theta=\theta_\dvec\\
		 -1 & \textnormal{otherwise}
		\end{array}
		\end{cases}.		
	\end{equation*}
	\end{enumerate}
	Finally, we set $k=\left(2+\binom{m}{m/2}+\binom{8m}{4m}\right)q + m$. 
	By setting $q\gg m$ (e.g., $q=2^{10m}$), candidate $a_0$ can get at least $k$ votes only if all receivers of type $\type_1$, $\type_2$, $\type_3$, $\type_4$ vote for her.

	\MyParagraph{Completeness.}
	Given a satisfiable assignment $\ass$ to the variables in $\varphi$, let $[\ass]_{T_j}\in\{0,1\}^{6m}$ be the vector specifying the variables assignment of each clause in $T_j$, and $[\ass]_{S_i}\in \{0,1\}^2m$ be the vector specifying the assignment of each variable belonging to $S_i$. 
	The sender has a signal for each $\dvec\in\{0,1\}^{7m}$. The set of signals is denoted by $\sset$, where $|\sset|=2^{7m}$, and a signal is denoted by $s_\dvec\in\sset$.
	We define a signaling scheme $\phi$ as follows.
	First, we set $\phi_{\theta_\dvec}(s_\dvec)=1$ for each $\theta_\dvec$.
	If $|T_j|<2m$ for some $j\in[m]$, we pad $[\zeta]_{T_j}$ with bits 0 util $|[\ass]_{T_j}|=6m$.
	Then, for each $T_j$, $\phi_{\theta_{(T_j,[\ass]_{T_j}\oplus \dvec_T)}}(s_\dvec)=1/2^m$.
	For each $i\in[8m]$, set $\phi_{\theta_{(i,\ell\oplus\dvec_S)}}=1/2^{6m}$, where $\ell=(e([\zeta]_{S_1})_i,\ldots, e([\zeta]_{S_m})_i)$.
	First, we prove that the signaling scheme is consistent.
	For each state $\theta_{(T_j,\ansDue)}$, it holds that
	\[
	\sum_{s_\dvec\in\sset}\phi_{\theta_{(T_j,\ansDue)}}(s_\dvec)=\frac{1}{2^m} |\{\dvec: [\ass]_{T_j}\oplus \dvec_T=\ansDue\}|=1,
	\]
	and, for each $\theta_{(i,\ell)}$, the following holds:
	\[
	\sum_{s_\dvec\in\sset}\phi_{\theta_{(i,\ell)}}(s_\dvec)=\frac{1}{2^{6m}}|\{\dvec: (e([\ass]_{S_1})_i,\ldots, e([\ass]_{S_m})_i\oplus \dvec_S=\ell \}|=1.
	\] 
	
	Now, we show that there exist at least $k$ voters that will choose $a_0$.
	Let $\pvec\in\Delta_\Theta$ be the posterior induced by a signal $s_\dvec$.
	All receivers of type $\type_1$ choose $a_0$ since it holds:
	\[\sum_{(T_j,\ansDue)}p_{\theta_{(T_j,\ansDue)}}=
	\frac{\sum_{(T_j,\ansDue)}\mu_{\theta_{(T_j,\ansDue)}}\phi_{\theta_{(T_j,\ansDue)}}(s_\dvec)}{\sum_{\theta\in\Theta}\mu_\theta\phi_\theta(s_\dvec)}=
	\frac{1}{2^{2+7m}}\left(\frac{1}{2^{1+7m}}+\frac{1}{2^{2+7m}}+\frac{1}{2^{2+7m}}\right)^{-1}=\frac{1}{4}.\]
	Analogously, all receivers of type $\type_2$ select $a_0$.
	For each $T_j$, it holds $\sum_{\ansDue} p_{\theta_{(T_j,\ansDue)}}=1/4m$.
	Then, for each subset $\tcal\subseteq\{T_j\}_{j\in[m]}$ of cardinality $m/2$, $\sum_{T_j\in\tcal,\ansDue}p_{\theta_{(T_j,\ansDue)}}=m/2\cdot1/4m=1/8$.
	Therefore, each receiver of type $\type_3$ chooses $a_0$.
	An analogous argument holds for receivers of type $\type_4$.
	
	%Finally, for each $S_i$, let $\ansUno^i \in \ansUnoSet$ be such that $e([\ass]_{S_i} \oplus d_{S} )_j= e(\ansUno^i)_j$ for all $j\in [8m]$.
	%Finally, for each $S_i$, there is at least one receiver $(\type_5,S_i,\ansUno,\dvec)$ who chooses $a_0$. 
	%
	%Let $\tilde \dvec$ be such that $e([\ass]_{S_i})_j\oplus \tilde d_{S,j} = e(\ansUno)_j$ for all $j\in [8m]$.
	%
	Finally, we show that, for each $S_i$, the receiver $(\type_5,S_i,[\ass]_{S_i},\dvec)$ chooses $a_0$. 
	Receiver $(\type_5,S_i,[\ass]_{S_i}, \dvec)$ has the following expected utility:
	\[
	\frac{1}{2}p_{\theta_{\dvec}}-\frac{1}{2} \sum_{(T_j,\ansDue)}p_{\theta_{(T_j,\ansDue)}}-\frac{1}{2}\sum_{(i',\ell)} p_{\theta_{(i',\ell)}}=0
	\]
	since, for each $p_{(T_j,\ansDue)}>0$, $\ansDue \oplus \dvec_T= [\ass]_{T_j}\oplus \dvec_T \oplus \dvec_T=[\ass]_{T_j}$ and $\ver(S_i,T_j,[\ass]_{S_i}, \ansDue \oplus\dvec_T)=\ver(S_i,T_j,[\ass]_{S_i}, [\ass]_{T_j}) =1$ for each $T_j$.
	Moreover, for each $p_{(\theta_{i',l})}>0$, $[l \oplus d_S]_{i}=e([\ass]_{S_i})_{i'} \oplus d_{S,i} \oplus d_{S,i}=e([\ass]_{S_i})_{i'}$.
	This concludes the completeness section. \footnote{
%		For easy of exposition, we use the indirect signals $s_d$. We can substitute each signal $s_d$ with a direct signal that recommends $a_0$ to all receivers with $\sum_\theta p_\theta u^r_\theta \ge 0$ and $a_1$ to all the other receivers, where $p$ is the posterior induced by signal $s_d$.
		%
		To simplify the presentation, we employed indirect signals of type $s_\dvec$. 
		However, it is possible to construct an equivalent direct signaling scheme. 
		Let $\pvec^\dvec\in\Delta_\Theta$ be the posterior induced by $s_\dvec$.
		Then, it is enough to substitute each $s_\dvec$ with a direct signal recommending $a_0$ to all receivers such that $\sum_\theta p^\dvec_\theta u_\theta^r\geq 0$, and $a_1$ to all the others.
}

	\MyParagraph{Soundness.}
	We prove that, if $\omega(\F_\varphi)\le 1-\delta$, there does not exists a posterior in which $a_0$ is chosen by at least $k$ receivers, thus implying that the sender's utility is equal to 0.
	Now, suppose, towards a contradiction, that there exists a posterior $\pvec$ such that at least $k$ receivers select $a_0$.
	Let $\gamma\defeq \sum_{(T_j,\ansDue)} p_{\theta_{(T_j,\ansDue)}} + \sum_{(i,\ell)} p_{\theta_{(i,\ell)}}$.
	Since all voters of types $\type_1$ and $\type_2$ vote for $a_0$, it holds that $\sum_{(T_j,\ansDue)} p_{\theta_{(T_j,\ansDue)}} \ge \frac{1}{4}-\epsilon$ and  $\sum_{(i,\ell)} p_{\theta_{(i,\ell)}} \ge \frac{1}{4}-\epsilon$.
	Moreover, since at least a receiver $(\type_5,S_i,\ansUno,\dvec)$ must play $a_0$, there exists a $\dvec\in\{0,1\}^{7m}$ and a state $\theta_\dvec$ with $p_{\theta_\dvec} \ge \frac{1}{2}-\epsilon$.
	This implies that $\frac{1}{2}-2\epsilon\le \gamma\le \frac{1}{2}+\epsilon$.
	
	Consider the reduction to $\epsilon'$-\textsc{MFS}, with $\epsilon'=\rho/52 d$ (Theorem \ref{th:mssi_lower}).
	Let $x_{(T_j,\ansDue)}=p_{\theta_{(T_j,\ansDue\oplus \dvec_T)}}/\gamma$, $x_{(i,\ell)}= p_{\theta_{(i,\ell \oplus \dvec_{S})}}/{\gamma}$, and $\epsilon=\epsilon'/30$.
	All rows of type $\type_1$ of $\epsilon'$-\textsc{\textsc{MFS}} are such that \[w_{\type_1}=\frac{1}{\gamma}\left(\sum_{(T_j,\ansDue)}p_{\theta_{(T_j,\ansDue)}}-\sum_{(i,l)} p_{\theta_{(i,l)}}\right)\ge -\frac{3\epsilon}{\gamma}\ge -9\epsilon \ge- \epsilon'.\]
	All voters of type $\type_3$ choose $a_0$. Then, for all $\mathcal{T}\subseteq \{T_j\}_{j \in [m]}$ of cardinality $m/2$, it holds:
	\[\sum_{(T_j,\ansDue): T_j \in \mathcal{T}} p_{\theta_{(T_j,\ansDue)}}- \sum_{(T_j,\ansDue): T_j \notin \mathcal{T}} p_{\theta_{(T_j,\ansDue)}} \ge-\epsilon.\] 
	Then, all rows of type $\type_2$ of \mssiPrime~are such that: 
	\[w_{(\type_2,\tcal)}= \frac{1}{\gamma}\left( \sum_{(T_j,\ansDue): T_j \in \mathcal{T}} p_{\theta_{(T_j,\ansDue)}}- \sum_{(T_j,\ansDue): T_j \notin \mathcal{T}} p_{\theta_{(T_j,\ansDue)}}\right)\ge -\frac{\epsilon} {\gamma}\ge -3\epsilon\ge-\epsilon'.\]
	A similar argument proves that all rows of type $\type_3$ of \mssiPrime~have $w_{(\type_3,\mathcal{I})}\ge -\epsilon'$.
	
	To conclude the proof, we prove that, for each voter $(\type_5,S_i,\ansUno,\dvec)$ that votes for $a_0$, the corresponding row $(\type_4,S_i,\ansUno)$ of \mssiPrime~is such that $w_{(\type_4,S_i,\ansUno)}\ge -\epsilon'$. 
	Let $\gamma'\defeq\sum_{(T_j,\ansDue):\ver(S_i,T_j,\ansUno,\ansDue) =1} x_{(T_j,\ansDue)} $ and $\gamma''\defeq\sum_{(i',\ell):e(\ansUno)_{i'}=\ell_i} x_{(i',\ell)}$.
	First, we have that $\gamma'\ge 1/4-7\epsilon$. If this did not hold, we would have 
	$$\sum_{\theta}p_\theta u_\theta^{(\type_5,S_i,\ansUno,\dvec)}<-\frac{1}{2}(1/4-\epsilon)-\frac{1}{2}(1/4-7\epsilon)-6\epsilon+\frac{1}{2}(1/2+2\epsilon)=\epsilon.$$
	Similarly, $\gamma'' \ge 1/4-7\epsilon$.
	Hence
	\begin{equation*}
	\begin{aligned}
	w_{(\type_4,S_i,\ansUno)} &= - \frac{1}{2} \gamma' + \frac{1}{2} \gamma'' - (1 -\gamma'-\gamma'') =\\
	&= 	\frac{1}{2\gamma}\left(\sum_{(T_j,\ansDue):\ver(S_i,T_j,\ansUno,\ansDue)=1}p_{\theta_{(T_j,\ansDue\oplus\dvec_T)}}
	+3\sum_{(i',\ell):e(\ansUno)_{i'}=\ell_i} p_{\theta_{(i',\ell\oplus\dvec_S)}}\right)-1\ge \\
	&\ge \frac{2(1/4-7\epsilon)}{1/2+\epsilon}-1 \ge  -30\epsilon=-\epsilon'.
	\end{aligned}
	\end{equation*}
	Thus, there exists a probability vector $\xvec$ for \mssiPrime~in which at least $k$ rows satisfy the \mssiPrime{} condition (Equation~\ref{eq:cond2_mssi_reduction}), which is in contradiction with $\omega(\F_\varphi)\leq1-\delta$.
	This concludes the proof.
\end{proof}

\section{A quasi-polynomial time algorithm for $(\alpha,\epsilon)$-persuasion}

In this section, we prove that our hardness result (Theorem~\ref{th:main_hardness}) is tight by devising a bi-criteria approximation algorithm.
Our result extends the results by~\citet{cheng2015mixture} and~\citet{xu2019tractability} for signaling problems with binary action spaces.
Indeed, it encompasses scenarios with an arbitrary number of actions and state-dependent sender's utility functions.

In order to prove our result, we need some further machinery. 
Let $\Z^r\defeq 2^{\A^r}$ be the power set of $\A^r$.
Then, $\Z\defeq \times_{r\in\rec}\Z^r$ is the set of tuples specifying a subset of $\A^r$ for each receiver $r$.
For a given probability distribution over the states of nature, we are interested in determining the set of best responses of each receiver $r$, i.e., the subset of $\A^r$ maximizing her expected utility. 
Formally, we have the following.
\begin{definition}[\br-set]\label{def:br_set}
	Given $\pvec\in\Delta_\Theta$, the {\em best-response set} (\br-set) $\brset_\pvec\defeq (Z^1,\ldots,Z^n)\in \Z$ is such that
	\[
	Z^r = \argmax_{a\in\A^r}\sum_{\theta\in\Theta}p_\theta u_\theta^r(a)\qquad\textnormal{ for each } r\in\rec.
	\]
\end{definition}

Similarly, we define a notion of $\epsilon$-\br-set which comprises $\epsilon$-approximate best responses to a given distribution over the states of nature.
\begin{definition}[$\epsilon$-\br-set]\label{def:eps_br_set}
	Given $\pvec\in\Delta_\Theta$, the {\em $\epsilon$-best-response set} ($\epsilon$-\br-set) $\brset_{\pvec,\epsilon}\defeq (Z^1,\ldots,Z^n)\in \Z$ is such that, for each $r\in\rec$, action $a$ belongs to $Z^r$ if and only if
	\[
	\sum_{\theta\in\Theta}p_\theta u_\theta^r(a)\geq \sum_{\theta\in\Theta}p_\theta u_\theta^r(a')-\epsilon \qquad\textnormal{ for each }a'\in\A^r.
	\]
\end{definition}

We introduce a suitable notion of {\em approximability} of the sender's objective function.
Our notion of {\em $\alpha$-approximable function} is a generalization of~\citet[Definition 4.5]{xu2019tractability} to the setting of arbitrary action spaces and state-dependent sender's utility functions.
\begin{definition}[$\alpha$-Approximability]\label{def:alfa_approx}
	Let $f\defeq\{f_\theta\}_{\theta\in\Theta}$ be a set of functions $f_\theta:\A\to[0,1]$.
	We say that $f$ is $\alpha$-approximable if there exists a function $g:\Delta_\Theta\times\Z\to\A$ computable in polynomial time such that, for all $\pvec\in\Delta_\Theta$ and $Z\in\Z$, it holds: $\avec=g(\pvec,Z)$, $\avec\in Z$ and 
	\[
	\sum_{\theta\in\Theta}p_\theta f_\theta(\avec)\geq \alpha \max_{\avec^\ast\in Z}\sum_{\theta\in\Theta} p_\theta f_\theta(\avec^\ast).
	\]
\end{definition}
The $\alpha$-approximability assumption is natural since otherwise it would be intractable even to evaluate the sender's objective value. 
When $f$ is $\alpha$-approximable, it is possible to find an approximation of the optimal receivers' tie breaking when they are constrained to select actions profiles in $Z$.

% can be expressed as a convex combination of s-uniform posteriors
% 

We now provide an algorithm which computes in quasi-polynomial time, for any $\alpha$-approximable $f$, a bi-criteria approximation of the optimal solution with an approximation  on the objective value arbitrarily close to $\alpha$.
When $f$ is $1$-approximate our result yields a bi-criteria QPTAS for the problem.
The key idea is showing  that an optimal signaling scheme can be approximated by a convex combination of suitable $k$-uniform posteriors.
Let $\nAct\defeq \max_{r\in\rec}\nAct_r$, $\nrec\defeq |\rec|$, and $d\defeq |\Theta|$.

\qptas
\begin{proof}
	We show that there exists a $\textnormal{poly}\left(d^{\frac{\log(\nrec\nAct/\delta)}{\epsilon^2}}\right)$ algorithm that computes the given approximation. 
	Let $k=\frac{32\log(4 \nrec\nAct/\delta)}{\epsilon^2}$ and $\K\subset\Delta_\Theta$ be the set of $k$-uniform distributions over $\Theta$ (Def.~\ref{def:k_distribution}).
	We prove that all posteriors $\pvec^\ast\in\Delta_\Theta$ can be decomposed as a convex combination of $k$-uniform posteriors without lowering too much the sender's expected utility.
	Formally, each posterior $\pvec^\ast\in\Delta_\Theta$ can be written as $\pvec^\ast=\sum_{\pvec\in\K}\gamma_{\pvec}\pvec$, with $\gamma\in\Delta_\K$ such that 
	\[
	\sum_{\pvec\in\K}\gamma_{\pvec}\sum_{\theta\in\Theta} p_\theta f_\theta(g(p,\brset_\epsilon(p)))\geq \alpha(1-\delta) \max_{\avec^\ast\in\brset(\pvec^\ast)}\sum_{\theta\in \Theta}p^\ast_\theta f_\theta(\avec^\ast).
	\]
	Let $\tilde\gamma\in\K$ be the empirical distribution of $k$ i.i.d. samples from $\pvec^\ast$, where each $\theta$ has probability $p^\ast_\theta$ of being sampled. 
	Therefore, the vector $\tilde\gamma$ is a random variable supported on $k$-uniform posteriors with expectation $\pvec^\ast$. 
	Moreover, let $\gamma\in\Delta_{\K}$ be a probability distribution such as, for each $\pvec\in\K$, $\gamma_\pvec\defeq \Pr(\tilde\gamma=\pvec)$.
	For a each $\gamma\in\Delta_\K$ and $\pvec\in\K$, we define by $\gamma_\pvec^{(\theta,i)}$ the conditional probability of having observed posterior $\pvec$, given that the posterior must assign probability $i/k$ to state $\theta$.
	Formally, for each $\pvec\in\K$, if $p_\theta=i/k$ we have $\gamma_\pvec^{(\theta,i)}=\gamma_\pvec/\sum_{\pvec':p'_\theta=i/k}\gamma_{\pvec'}$, and $\gamma_\pvec^{(\theta,i)}=0$ otherwise.
	The random variable $\tilde\gamma^{(\theta,i)}\in\K$ is such that, for each $\pvec\in\K$, $\Pr(\tilde\gamma^{(\theta,i)}=\pvec)=\gamma^{(\theta,i)}_\pvec$.
	Finally, let $\Pcal\subseteq\K$ be the set of posteriors such that $\pvec\in\Pcal$ if and only if $|\sum_{\theta}p_\theta u^r_\theta(a)-\sum_{\theta}p^\ast_\theta u_\theta^r(a)|\leq \frac{\epsilon}{2}$ for each $r \in \rec$ and $a\in\A^r$.
	
	We prove the following intermediate result.
	\begin{restatable}{lemma}{bicriteriaOne}\label{lemma:bicriteria_1}
		Given $\pvec^\ast\in \Delta_\Theta$, for each $\theta\in\Theta$ and for each $i\in[k]$ such that $|	i/k-p^\ast_\theta|\leq \epsilon/4$, it holds:
		\[
		\sum_{\pvec\in\Pcal:p_\theta=i/k}\gamma_\pvec\geq \left(1-\frac{\delta}{2}\right)\sum_{\pvec\in\K:p_\theta=i/k}\gamma_\pvec,
		\]
		where $\gamma$ is the distribution of $k$ i.i.d samples from $\pvec^\ast$.
	\end{restatable}
	\begin{proof}
		Fix $\bar \theta\in\Theta$ and $i\in[k]$ with $|i/k-p^*_{\bar \theta}|\le \epsilon/4$. 
		Then, for each $r\in\rec$ and $a\in\A^r$, let $\tilde t^r_a\defeq \sum_{\theta}\tilde\gamma_\theta^{(\bar\theta,i)} u^r_\theta(a)$ and $t_a^r\defeq \sum_{\theta} p^\ast_\theta u_\theta^r(a)$.
		First, we show that $|\Expec[\tilde t_a^r]-t_a^r|\leq\epsilon/4$.
		Equivalently, 
		$|\sum_{\theta}u_\theta^r(a)\left(\Expec[\tilde\gamma_\theta^{(\bar\theta,i)}]-p^\ast_\theta\right)|\leq \epsilon/4$.
		Assume $i/k\geq p^\ast_{\bar\theta}$.
		Then, 
		\begin{subequations}\label{eq:risultatino}
			\begin{align}
			\sum_{\theta}|\Expec[\tilde\gamma^{(\bar\theta,i)}_\theta]-p^\ast_\theta|=&\frac{i}{k}-p^\ast_{\bar\theta}+\sum_{\theta\neq\bar\theta}\left(p^\ast_\theta - \frac{p^\ast_\theta}{\sum_{\theta'\neq\bar\theta}p^\ast_{\theta'}}\cdot\left(1-\frac{i}{k}\right) \right)\leq\\
			\leq & \frac{\epsilon}{4}+1-p^\ast_{\bar\theta}-1+\frac{i}{k}\leq \frac{\epsilon}{2}.
			\end{align}	
		\end{subequations}
		Analogously, if $i/k\leq p^\ast_{\bar\theta}$, we get that $\sum_{\theta}|\Expec[\tilde\gamma^{(\bar\theta,i)}_\theta]-p^\ast_\theta|\leq \epsilon/2$.
		Let $M_1\defeq \left\{\theta\in\Theta\mid \Expec[\tilde\gamma_\theta^{(\bar\theta,i)}]-p^\ast_\theta\geq 0\right\}$, and $M_2\defeq \Theta\setminus  M_1$.
		Then, 
		\begin{align*}
		\sum_{\theta}u_\theta^r(a)\left(\Expec[\tilde\gamma^{(\bar\theta,i)}_\theta]-p^\ast_\theta\right)=&\sum_{\theta\in M_1}u_\theta^r(a)\left(\Expec[\tilde\gamma^{(\bar\theta,i)}_\theta]-p^\ast_\theta\right)+\sum_{\theta\in M_2}u_\theta^r(a)\left(\Expec[\tilde\gamma^{(\bar\theta,i)}_\theta]-p^\ast_\theta\right)\leq \frac{\epsilon}{4},
		\end{align*}
		where we use $\sum_{\theta\in M_2}u_\theta^r(a)\left(\Expec[\tilde\gamma^{(\bar\theta,i)}_\theta]-p^\ast_\theta\right)\le0$ and $\sum_{\theta\in M_1}u_\theta^r(a)\left(\Expec[\tilde\gamma^{(\bar\theta,i)}]-p^\ast_\theta\right)\le \epsilon/4$ (by Equation~\ref{eq:risultatino}).
		Analogously, it is possible to show that $\sum_{\theta}u_\theta^r(a)\left(\Expec[\tilde\gamma^{(\bar\theta,i)}_\theta]-p^\ast_\theta\right)\geq -\epsilon/4$.
		
		Then, $\Pr(|t_a^r-\tilde t_a^r|\geq \epsilon/2)\leq \Pr(|\tilde t_a^r-\Expec[\tilde t_a^r]|\geq \epsilon/4)$. 
		Moreover, by the Hoeffding's inequality we have that, for each $r\in\rec$ and $a\in\A^r$,
		\begin{equation*}\label{eq:lemma_hoeffding}
		\Pr(|\tilde t^r_a - \Expec[\tilde t^r_a]|\geq \epsilon/4)\leq 2e^{-2k(\frac{\epsilon}{4})^2}=2e^{\frac{-4\epsilon^2\log(4\nrec\nAct/\delta)}{\epsilon^2}}=2\left(\frac{\delta}{4\nrec\nAct}\right)^4\leq
		\frac{\delta}{2\nrec\nAct} .
		\end{equation*}
		The union bound yields the following:
		\begin{align*}
		\Pr\left(\bigcap_{r\in\rec,a\in\A^r} |\tilde t_a^r-t_a^r|\leq \frac{\epsilon}{2}\right)\geq & 1-\sum_{r,a}\Pr\left(|\tilde t^r_a-t^r_a|\geq \frac{\epsilon}{2}\right)\geq\\
		\geq &1-\sum_{r,a}\Pr\left(|\tilde t^r_a-\Expec[\tilde t_a^r]|\geq \frac{\epsilon}{4}\right)=1-\frac{\delta}{2}.
		\end{align*}
		By the definition of $\Pcal$, this implies that $\Pr(\tilde\gamma^{(\bar \theta,i)}\in\Pcal)\geq 1-\delta/2$.
		Finally,
		\begin{align*}
		\sum_{\pvec \in \Pcal : p_{\bar{\theta}}=i/k} \gamma_p =& \Pr\left( \tilde{\gamma}_{\bar\theta}=\frac{i}{k}\right) \Pr \left( \tilde \gamma \in \Pcal \mid \tilde\gamma_{\bar\theta}=\frac{i}{k}\right)=\\=& 
		\Pr\left( \tilde{\gamma}_{\bar\theta}=\frac{i}{k}\right)\Pr\left(\tilde\gamma^{(\bar\theta,i)}\in\Pcal\right)\geq\\\geq&
		\left(1-\frac{\delta}{2}\right) \Pr\left( \tilde{\gamma}_{\bar\theta}=\frac{i}{k}\right)=\left(1-\frac{\delta}{2}\right) \sum_{\pvec \in \K : p_{\bar{\theta}}=i/k} \gamma_p.
		\end{align*}
	\end{proof}

	Then, we prove the following auxiliary lemma:
	\begin{restatable}{lemma}{bicriteriaDue}\label{lemma:bicriteria_2}
		Given $\pvec^\ast\in \Delta_\Theta$, for each $\theta \in \Theta$, it holds:
		\[
		\sum_{i : |i/k-\pvec^\ast_\theta|\ge \epsilon/4} \sum_{\pvec \in \K:p_\theta=i/k} \gamma_p \le \frac{\delta}{2} p^\ast_\theta,
		\]
		where $\gamma$ is the distribution of $k$ i.i.d samples from $\pvec^\ast$.
	\end{restatable}
	\begin{proof}
		The random variable $\tilde\gamma_{\theta}$ is drawn from a binomial distribution.
		Then, by Chernoff's bound
		\begin{subequations}
			\begin{align}
			\Pr\left(|\tilde{\gamma}_\theta-p^\ast_\theta|\ge \frac{\epsilon}{4} \right) \leq& 2 e^{-\frac{k\epsilon^2}{32p^\ast_\theta}}=2 e^{-\frac{32\log(4\nrec\nAct/\delta)}{32p^\ast_\theta}} =
			2 \left(\frac{\delta}{4\nrec\nAct}\right)^\frac{1}{p^\ast_\theta}\leq \label{eq:lemma_1}\\
			\le & 2 \left(\frac{\delta}{16}\right)^\frac{1}{p^\ast_\theta}= \label{eq:lemma_2}\\
			=& 2e^{\log\left(\frac{\delta}{16}\right)\frac{1}{p^\ast_\theta}} = 2 \left( e^{\frac{1}{p^\ast_\theta}} \right)^{\log\left(\frac{\delta}{16}\right)}\leq\label{eq:lemma_3}\\
			\le & 2 \left(\frac{1}{p^\ast_\theta} e\right)^{\log\left(\frac{\delta}{16}\right)}\leq \label{eq:lemma_4}\\ \le &  
			2 \left(\frac{1}{p^\ast_\theta}\right)^{-1} e^{\log\left(\frac{\delta}{16}\right)} \le \\
			\le & \frac{\delta}{2} p^\ast_\theta.\label{eq:lemma_5}
			\end{align}
		\end{subequations}
		We get from  \eqref{eq:lemma_1} to \eqref{eq:lemma_2} via the natural assumption of having at least 2 actions for each receiver (i.e., $\delta\geq 2$), and of having at least 2 receivers (i.e., $\nrec\geq 2$).
		In \eqref{eq:lemma_4} we are using $e^x\ge e x$.
		Then,
		\[
		\sum_{i : |i/k-\pvec^\ast_\theta|\ge \epsilon/4} \sum_{\pvec \in \K:p_\theta=i/k} \gamma_p =\Pr\left( |\tilde\gamma_\theta-p^\ast_\theta | \ge \frac{\epsilon}{4} \right)\leq \frac{\delta}{2}p_\theta^\ast,
		\]
		which concludes the proof of the lemma.
	\end{proof}
	
	Now we can prove that, given a $\pvec^\ast \in \Delta_\Theta$ and for each $\theta$, $\sum_{\pvec \in \Pcal} \gamma_p p_\theta\ge (1-\delta) p^*_\theta$.
	\begin{restatable}{lemma}{bicriteriaTre}\label{lm:bicriteria_3}
		Given a  $\pvec^\ast \in \Delta_\Theta$, for each $\theta\in\Theta$, it holds:
		\[\sum_{\pvec \in \Pcal} \gamma_p p_\theta\ge (1-\delta) p^*_\theta,	\]
		where $ \gamma$ is the distribution of $k$ i.i.d samples from $\pvec^\ast$.
	\end{restatable}
	\begin{proof}
		First, by restricting the set of posteriors, we have:
		\[
		\sum_{\pvec \in \Pcal} \gamma_\pvec p_\theta \geq
		\sum_{i:|i/k-\pvec^\ast_\theta|\le \epsilon/4} \frac{i}{k} \sum_{\pvec \in \Pcal:p_\theta=i/k} \gamma_p.
		\]
		By Lemma~\ref{lemma:bicriteria_1},
		\[
		\sum_{i:|i/k-\pvec^\ast_\theta|\le \epsilon/4} \frac{i}{k} \sum_{\pvec \in \Pcal:p_\theta=i/k} \gamma_p	\geq \sum_{i:|i/k-p^\ast_\theta|\le \epsilon/4} \frac{i}{k} \sum_{\pvec \in \K:p_\theta=i/k} \left(1-\frac{\delta}{2}\right) \gamma_\pvec .
		\]
		Finally,
		
		\begin{subequations}\label{eq:1}
			\begin{align*}
			\sum_{i:|i/k-p^\ast_\theta|\le \epsilon/4} \frac{i}{k} \sum_{\pvec \in \K:p_\theta=i/k} \left(1-\frac{\delta}{2}\right) \gamma_\pvec=&
			\left(1-\frac{\delta}{2}\right)  \sum_{i:|i/k-p^\ast_\theta|\le \epsilon/4} \frac{i}{k} \sum_{\pvec \in \K:p_\theta=i/k} \gamma_\pvec \ge\\ \geq &
			\left(1-\frac{\delta}{2}\right) \left(p^\ast_\theta -\sum_{i:|i/k-p^\ast_\theta|\ge \epsilon/4} \frac{i}{k}\sum_{\pvec \in \K:\pvec_\theta=i/k} \gamma_\pvec\right) \ge \\
			\geq &
			\left(1-\frac{\delta}{2}\right) \left(p^\ast_\theta -\sum_{i:|i/k-p^\ast_\theta|\ge \epsilon/4} \sum_{\pvec \in \K:\pvec_\theta=i/k} \gamma_\pvec\right) \ge \hspace{1cm} \textnormal{($i/k\leq 1$)}\\
			\geq & \left(1-\frac{\delta}{2}\right)^2 p^*_\theta \ge \hspace{5cm} 
			\textnormal{(by Lemma~\ref{lemma:bicriteria_2})}\\
			\geq & (1-\delta) p^*_\theta.
			\end{align*}
		\end{subequations}
		This concludes the proof of the lemma.
	\end{proof}

	We need to prove that all the posteriors in $\Pcal$ guarantee to the sender at least the same expected utility of $\pvec^\ast$.
	Formally, we prove that the $\epsilon$-\br-set of each $\pvec \in \Pcal$ contains the \br-set of $\pvec^\ast$.
	This is shown via the following lemma.
	\begin{restatable}{lemma}{bicriteriaQuattro}\label{lemma:bicriteria_4}
		Given $\pvec^\ast \in \Delta_{\Theta}$, for each $\pvec \in \Pcal$, it holds: $\brset(\pvec^\ast) \subseteq  \brset_\epsilon(\pvec)$. 
	\end{restatable}
	\begin{proof}
		Let $Z_1=\brset_\epsilon(\pvec)$ and $Z_2=\brset(\pvec^\ast)$. Suppose $a \in Z_2^r$. Then,  for all $a' \in \A^r$, 
		\[\sum_\theta p_\theta u^r_\theta(a) \ge \sum_\theta p^\ast_\theta u^r_\theta(a)-\frac{\epsilon}{2} \ge \sum_\theta p^\ast_\theta u^r_\theta(a')- \frac{\epsilon}{2} \ge \sum_\theta p_\theta u^r_\theta(a')- \epsilon. 
		\]
		Thus, $a \in Z_1^r$, which proves the lemma.
	\end{proof}

	Finally, we prove that we can represent each posterior $\pvec^\ast$ as a convex combination of $k$-uniform posteriors with a small loss in the sender's expected utility.
	For $\pvec \in \K$ and $Z \in \Z$, let $g^\ast: \Delta_\Theta\times\Z\to [0,1]$ be a function such that 
	$g^*(\pvec,Z)\defeq\max_{\avec \in Z} \sum_{\theta} p_\theta f_\theta(\avec)$.
	Given $\pvec^\ast\in\Delta_\Theta$, we are interested in bounding the difference in the sender's expected utility when $\pvec^\ast$ is approximated as a convex combination $\gamma$ of $k$-uniform posteriors, the sender exploits an $\alpha$-approximation of $f$, and she allows receivers for $\epsilon$-persuasive best-responses.
	Formally,
	
	\begin{restatable}{lemma}{lemmaFinale}
		Given a $\pvec^\ast \in \Delta_\Theta$, it holds:
		\[\sum_{\pvec \in \K} \gamma_\pvec \sum_{\theta} p_\theta f_\theta(g(\pvec,\brset_\epsilon(\pvec))) \ge f_\theta(g^\ast(\pvec^\ast,\brset(\pvec^\ast))),\]
		where $ \gamma$ is the distribution of $k$ i.i.d samples from $\pvec^\ast$.
	\end{restatable}
	\begin{proof}
		
		We prove the following:
		\begin{subequations}
			\begin{align*}
			&\sum_{\pvec \in \K} \gamma_\pvec \sum_{\theta} p_\theta f_\theta(g(\pvec,\brset_\epsilon(\pvec))) \ge \hspace{1.8cm}\textnormal{(Relaxed sender's expected util.)}\\
			\geq & \alpha \sum_{\pvec \in \K} \gamma_\pvec \sum_{\theta} p_\theta f_\theta(g^\ast(\pvec,\brset_\epsilon(\pvec))) \ge \hspace{4.3cm} \textnormal{(by Def.~\ref{def:alfa_approx})}\\
			\geq & \alpha \sum_{\pvec \in \Pcal} \gamma_\pvec \sum_{\theta} p_\theta f_\theta(g^\ast(\pvec,\brset_\epsilon(\pvec))) \ge \hspace{1cm}\textnormal{(By restricting the set of posteriors)}\\
			\geq & \alpha \sum_{\pvec \in \Pcal} \gamma_\pvec \sum_{\theta} p_\theta f_\theta(g^\ast(\pvec^\ast,\brset_\epsilon(\pvec))) \ge \hspace{3.4cm}\textnormal{(Optimality of $g^\ast$)}\\
			\geq & \alpha \sum_{\pvec \in \Pcal} \gamma_\pvec \sum_{\theta} p_\theta f_\theta(g^\ast(\pvec^\ast,\brset(\pvec^\ast))) \ge\hspace{3.7cm}\textnormal{(By Lemma~\ref{lemma:bicriteria_4})}\\
			\geq & \alpha (1-\delta)  \sum_{\theta} p_\theta^\ast f_\theta(g^\ast(\pvec^\ast,\brset(\pvec^\ast)))\hspace{4.1cm}\textnormal{(By Lemma~\ref{lm:bicriteria_3})}
			\end{align*}
		\end{subequations}
		This concludes the proof.
	\end{proof}

	Thus, we can restrict to posteriors in $\K$. Since there are $|\K|=\textnormal{poly}\left(d^{\frac{\log(\nrec \nAct/\epsilon)}{ \epsilon^2}}\right)$ posteriors, the following linear program (LP \ref{eq:lp1}) has $O(|\mathcal{K}|)$ variables and constraints and finds a $\alpha(1-\delta)$-approximation of the optimal signaling scheme:
	\begin{subequations}\label{eq:lp1}
		\begin{align}
		\max_{\gamma\in \Delta_\K} & \sum_{\pvec \in \K}  \gamma_\pvec \sum_{\theta \in \Theta} p_\theta f_\theta(g(\pvec,\mathcal{M}_\epsilon(	\pvec))) \label{lp1:obj}\\
		\textnormal{s.t.} &\sum_{p \in \K} \gamma_p p_\theta =\mu_\theta \qquad\forall \theta \in \Theta \label{lp1:cons1}
		\end{align}
	\end{subequations}
	
	Given the distribution on the $k$-uniform posteriors $\gamma$, we can construct a direct signaling scheme $\phi$ by setting, for each $\theta\in\Theta$ and $\avec\in\A$, 
	\[
	\phi_\theta(\avec)=\sum_{\pvec\in\K:\avec=g(\pvec,\mathcal{M}_\epsilon(\pvec))} \gamma_\pvec p_\theta.
	\]
	We showed that such a $\phi$ is $\alpha(1-\delta)$-approximate and $\epsilon$-persuasive, which are precisely our desiderata. 
	This concludes the proof.
\end{proof}
%

% Bibliography
\bibliographystyle{named}
\bibliography{biblio}

% Appendix

\end{document}